\documentclass[a4paper,UKenglish,cleveref, autoref, thm-restate]{lipics-v2021}
\usepackage{thm-restate}

\DeclareMathOperator*{\argmax}{arg\,max}

\title{Fault Tolerance in Euclidean Committee Selection} %TODO Please add

\author{Chinmay Sonar}{Department of Computer Science, University of California, Santa Barbara, USA}{csonar@cs.ucsb.edu}{https://orcid.org/0000-0002-1825-0097}{}%TODO mandatory, please use full name; only 1 author per \author macro; first two parameters are mandatory, other parameters can be empty. Please provide at least the name of the affiliation and the country. The full address is optional. Use additional curly braces to indicate the correct name splitting when the last name consists of multiple name parts.

\author{Subhash Suri}{Department of Computer Science, University of California, Santa Barbara, USA}{suri@cs.ucsb.edu}{}{}

\author{Jie Xue}{Department of Computer Science, New York University, Shanghai, China}{jiexue@nyu.edu}{}{}

\authorrunning{C. Sonar, S. Suri and J. Xue} %TODO mandatory. First: Use abbreviated first/middle names. Second (only in severe cases): Use first author plus 'et al.'

\Copyright{Jane Open Access and Joan R. Public} %TODO mandatory, please use full first names. LIPIcs license is "CC-BY";  http://creativecommons.org/licenses/by/3.0/

\ccsdesc[100]{{Theory of computation $\xrightarrow{}$ Randomness geometry and discrete structures $\xrightarrow{}$ Computational geometry}} %TODO mandatory: Please choose ACM 2012 classifications from https://dl.acm.org/ccs/ccs_flat.cfm 

\keywords{Multiwinner elections, Fault tolerance, Geometric Hitting Set, EPTAS} %TODO mandatory; please add comma-separated list of keywords

\category{} %optional, e.g. invited paper

\relatedversion{} %optional, e.g. full version hosted on arXiv, HAL, or other respository/website
\nolinenumbers %uncomment to disable line numbering

\newcommand{\defproblemgoal}[3]{
\vspace{2mm}
\noindent\fbox{
  \begin{minipage}{0.98\textwidth}
  \begin{tabular*}{\textwidth}{@{\extracolsep{\fill}}lr} #1 \\ \end{tabular*}
  {\bf{Input:}} #2  \\
  {\bf{Goal:}} #3
  \end{minipage}
}\vspace{2mm}}

\usepackage{subcaption}
\usepackage{algorithm}
\usepackage[noend]{algpseudocode}

\EventEditors{John Q. Open and Joan R. Access}
\EventNoEds{2}
\EventLongTitle{42nd Conference on Very Important Topics (CVIT 2016)}
\EventShortTitle{CVIT 2016}
\EventAcronym{CVIT}
\EventYear{2016}
\EventDate{December 24--27, 2016}
\EventLocation{Little Whinging, United Kingdom}
\EventLogo{}
\SeriesVolume{42}
\ArticleNo{23}
\begin{document}

\maketitle

\begin{abstract}
In the committee selection problem, the goal is to choose a subset of size $k$ from a set of candidates $C$ that collectively 
gives the best representation to a set of voters. We consider this problem in Euclidean $d$-space where each voter/candidate is a point and voters' preferences are implicitly represented by Euclidean distances to candidates.
We explore \emph{fault-tolerance} in committee selection and study the following three variants:
(1) given a committee and a set of $f$ failing candidates, find their optimal replacement;
(2) compute the worst-case replacement score for a given committee under failure of $f$ candidates; and
(3) design a committee with the best replacement score under worst-case failures.
The score of a committee is determined using the well-known (min-max) Chamberlin-Courant rule: minimize the maximum distance between any voter and its closest candidate in the committee.  Our main results include the following: 
(1) in one dimension, all three problems can be solved in polynomial time;
(2) in dimension $d \geq 2$, all three problems are NP-hard; and 
(3) all three problems admit a constant-factor approximation in any fixed dimension, and the optimal committee
problem has an FPT bicriterion approximation.
\end{abstract}

\section{Introduction and Problem Statement}

We consider the computational complexity of adding fault tolerance into spatial voting. In spatial voting~\cite{davis1970expository,arrow1990advances,sonar2022}, 
the voters and the candidates are both modeled as points in some $d$-dimensional space, where each dimension represents 
an independent \emph{policy issue} that is important for the election, and each voter's preference among the candidates 
is implicitly encoded by a distance function. For example, in the simplest $1$-dimensional setting, voters and candidates 
are points on a line indicating their real-valued preference on a single issue.
The specific setting for our work is \emph{multiwinner} spatial elections, also called \emph{committee selection},
in $d$ dimensions where we have a set $V$ of $n$ voters, a set $C$ of $m$ candidates, and a committee size (integer) $k$.
The goal is to choose a subset of $k$ candidates, called the \emph{winning committee}, that collectively best represents the preferences of all the voters~\cite{elkind2017properties,faliszewski2017multiwinner,faliszewski2019committee}.

One aspect of committee selection that appears not to have been investigated is \emph{fault tolerance}, that is, how robust a chosen committee is against the possibility that some of the winning members may default.
Committee selection problems model a number of applications in the social sciences and in computer science where such defaults are not uncommon,
such as democratic elections, staff hiring, choosing public projects, locations of public facilities, jury selection, cache management, etc.~\cite{lu2011budgeted,goel2019knapsack,skowron2016finding,betzler2013computation,bredereck2021coalitional,maiyya2018database,friedman2019fair}.
In this paper, we are particularly interested in designing algorithms to address questions of the following kind: If some of the  winning members default, how badly does this affect the overall score of the committee? 
Or, how much does the committee score suffer if a \emph{worst-case} subset of size $f$ defaults?
Finally, can we \emph{proactively} choose a committee in such a way that it can tolerate up to $f$ faults
with the minimum possible score degradation? 
We begin by formalizing these problems more precisely and then describing our results.

Suppose $V = \{v_1, \ldots, v_n \}$ is a set of $n$ voters and $C = \{c_1, \ldots, c_m \}$ is a set of $m$ candidates,
modeled as points in $d$-dimensional Euclidean space. (We occasionally call the tuple $(V, C)$ an election $E$.)
Given a positive integer $k$, we want to elect $k$ candidates, called the \emph{committee}, using the well-known 
Chamberlin-Courant voting rule~\cite{chamberlin1983representative}.
This rule assigns a score to each committee as follows. 
Let $T \subseteq C$ be a committee. For each voter $v$, the score of $T$ for $v$ is defined as $\sigma(v, T) = \min_{c \in T} d(v, c)$,
namely, the distance from $v$ to its closest candidate in $T$.\footnote{Originally, Chamberlin and Courant \cite{chamberlin1983representative} defined a voting rule on Borda scores (also known as Borda-CC). In this paper, similarly to \cite{betzler2013computation}, we study a min-max version of this rule on a more general scoring function, which in our case is based on voter-candidate distances.}
The \emph{score} of the committee $T$ is defined as $\sigma(T) = \max_{v \in V} \sigma(v,T)$, namely, the largest distance between any voter and its closest neighbor in $T$. (In facility location parlance, this is  the well-known $k$-center problem.)

The fault tolerance of a committee is parameterized by a positive integer $f$, which is the upper bound on the number of candidates 
that can fail.\footnote{%
	In our work, we will allow any subset of size $f$ from $C$ to fail, so the faults can also include candidates not 
	in the selected committee $T$. This only makes the problem harder because the adversary can always limit the faults to $T$, 
	and elimination of candidates from $C \setminus T$ makes finding replacements for failing committee members more difficult.}
Throughout the paper, we use the notation $J$ to denote a failing set of candidates. We are allowed to replace the failing 
members of $J$ with any set of at most $|T \cap J|$ candidates from $C \setminus J$. We often denote this set of 
\emph{replacement} candidates by $R$. However, \emph{we must keep all the non-failing members of $T$ in the committee} ---
that is, the replacement committee is the set $(T\setminus J) \cup R$ --- and throughout the paper our goal is to optimize this
committee's score, namely $\sigma((T\setminus J) \cup R)$. 

We consider the following three versions of fault-tolerant committee selection, presented in increasing order of complexity.
The first problem is the simplest: given a committee and a failing set, find the best replacement committee.

\defproblemgoal{{\textbf{Optimal Replacement Problem} (ORP)}}{An election $E=(V,C)$, a committee $T \subseteq C$ and a failing set $J \subseteq C$.}{Find a replacement set $R \subseteq C\setminus J$ of size at most $|T \cap J|$ minimizing $\sigma((T\setminus J) \cup R)$.}

Our second problem is to quantify the fault tolerance of a given committee $T$ over worst-case faults.
That is, what is the largest score of $T$'s replacement when a worst-case subset of $f$ faults occur?
We introduce the following notation as $T$'s measure of $j$-fault-tolerance, for any $0 \leq j \leq f$:
	$\sigma_j (T) ~~=~~ \max_{J \subseteq C~ s.t. ~ |J| \leq j}  \sigma( (T\setminus J) \cup R),$
where $R$ is an optimal replacement set with size at most $|T \cap J|$.  We want to compute $\sigma_f (T)$. 
Occasionally, we also use the notation $\sigma_0(T)$ for the no-fault score of $T$, namely $\sigma(T)$.

\defproblemgoal{\textbf{Fault-Tolerance Score} (FTS)}{An election $E=(V,C)$, a committee $T \subseteq C$ and a fault-tolerance parameter $f$.}{Compute $\sigma_f(T)$.}

Our third and final problem is to compute a committee with optimal fault-tolerance score.

\defproblemgoal{\textbf{Optimal Fault-Tolerant Committee} (OFTC)}{An election $E=(V,C)$, a committee size $k$ and a fault-tolerance parameter $f$.}{Find $T \subseteq C$ of size at most $k$ minimizing $\sigma_f(T)$.}

\subsection{Our Results}

We first show that even in one dimension, fault-tolerant committee problems are nontrivial.  In particular, while the Optimal Replacement Problem (ORP) is easily solved by a simple greedy algorithm, the other two problems, Fault-Tolerance 
Score (FTS) and Optimal Fault-Tolerant Committee (OFTC), do not appear to be easy. Our main result in one dimension
is the design of efficient dynamic-programming-based algorithms for these two problems. 
Along the way, we solve a \emph{fault-tolerant} Hitting Set problem for points and unit intervals, which may be of independent interest.

In two dimensions and higher, OFTC is NP-hard because of its close connection to the $k$-center 
problem. However, we show that even the seemingly simpler problem of optimal replacement (ORP) is also NP-hard. 
Our main results include a constant-factor approximation for all three problems in any fixed dimension (in fact, in any metric space), 
as well as a novel bicriterion FPT approximation via an EPTAS whose running time has the form 
$f(\epsilon)n^{\mathcal{O}(1)}$. For ease of reference, we show these results in the following table.

% \usepackage{multirow}
\begin{table}[h!]
\centering
\begin{tabular}{c|c|ccc|}
\cline{2-5}
\multicolumn{1}{l|}{}               & \multirow{2}{*}{\textbf{\begin{tabular}[c]{@{}c@{}}One-dimensional\\ instances\end{tabular}}} & \multicolumn{3}{c|}{\textbf{Dimension $d \geq 2$}}                                                                                                                                                                                                                        \\ \cline{3-5} 
\multicolumn{1}{l|}{\textbf{}}      &                                                                                               & \multicolumn{1}{c|}{\textbf{Complexity}}                                          & \multicolumn{1}{c|}{\textbf{Approximation}}                                                                          & \textbf{Bounded $f$}                                            \\ \hline
\multicolumn{1}{|c|}{\textbf{ORP}}  & \begin{tabular}[c]{@{}c@{}}P\\ (Theorem \ref{thm:1d-ORP})\end{tabular}          & \multicolumn{1}{c|}{\begin{tabular}[c]{@{}c@{}}NP-hard\\ (Theorem~\ref{thm-hardness-2d})\end{tabular}}                                   & \multicolumn{1}{c|}{\begin{tabular}[c]{@{}c@{}}3-approx.\\ (Lemma~\ref{lem:ORP-2d-3-approx})\end{tabular}}                           & \multicolumn{1}{c|}{\begin{tabular}[c]{@{}c@{}}P\\ (Section~\ref{subsec:bounded-f})\end{tabular}}                                                              \\ \hline
\multicolumn{1}{|c|}{\textbf{FTS}}  & \begin{tabular}[c]{@{}c@{}}P\\ (Theorem \ref{thm:1d-FTS})\end{tabular}                                       & \multicolumn{1}{c|}{\begin{tabular}[c]{@{}c@{}}NP-hard\\ (Theorem \ref{thm-hardness-2d})\end{tabular}}                                                      & \multicolumn{1}{c|}{\begin{tabular}[c]{@{}c@{}}3-approx.\\ (Lemma \ref{lem:FTS-2d-3-approx})\end{tabular}}                                    & \multicolumn{1}{c|}{\begin{tabular}[c]{@{}c@{}}P\\ (Section~\ref{subsec:bounded-f})\end{tabular}}                                                              \\ \hline
\multicolumn{1}{|c|}{\textbf{OFTC}} & \begin{tabular}[c]{@{}c@{}}P\\ (Theorem \ref{thm:1d-OFTC})\end{tabular}                                       & \multicolumn{1}{c|}{\begin{tabular}[c]{@{}c@{}}NP-hard\\ (Theorem~\ref{thm-hardness-2d})\end{tabular}} & \multicolumn{1}{c|}{\begin{tabular}[c]{@{}c@{}}5-approx.\\ (Lemma~\ref{lem:OFTC-2d-5-approx})\\ Bicriterion-EPTAS\\ (Theorem \ref{thm:OFTC-2d-EPTAS})\end{tabular}} & \multicolumn{1}{c|}{\begin{tabular}[c]{@{}c@{}}NP-hard\\ (Section~\ref{subsec:bounded-f})\\ 3-approx. \\ (Theorem \ref{thm:OFTC-2d-3-approx})\end{tabular}} \\ \hline
\end{tabular}
\caption{Summary of our results.}
\end{table}

%\vspace{-1cm}
\subsection{Related Work}

To the best of our knowledge, the issue of fault tolerance in committee selection has not been studied in
voting literature --- their primary focus is on protocols and algorithms for choosing candidates~\cite{faliszewski2017multiwinner, faliszewski2019committee, munagala2021optimal, skowron2015achieving, sonar2021complexity, elkind2017multiwinner}. However, the following two lines of work consider some related issues. First, in the ``unavailable candidate model''~\cite{lu2010unavailable, grivet2021preference} 
the goal is to choose a \emph{single winner} with maximum expected score when candidates fail according to a given probability distribution; in contrast, we consider multiwinner elections under worst-case faults.
In the second line of work, a set of election control problems are considered where candidates are added \cite{FaliszewskiHHR09} or deleted \cite{hemaspaandra2007anyone} to change the outcome of the election. In this setting, the candidate set is modified to obtain a favorable election outcome, which is a rather different problem than ours.

In the facility-location research, there has been prior work on adding fault tolerance to $k$-center or 
$k$-median solutions~\cite{chaudhuri1998p, krumke1995generalization, khuller2000fault, swamy2008fault, hajiaghayi2016constant}, but the main approach there is to assign each user (voter) to multiple facilities (candidates). 
In particular, the ``$p$-neighbor $k$-center'' framework~\cite{chaudhuri1998p} minimizes the maximum
distance between a user and its $p$th center as a way to protect against $p-1$ faults. This formulation, however, differs from our optimal fault-tolerant committee problem (OFTC) because in our setting the replacement candidates are chosen \emph{after} failing candidates are 
announced. Therefore, in the OFTC problem, the designer does not have to \emph{simultaneously} allocate 
$p$ neighbors for all the voters. Furthermore, to the best of our knowledge, neither of our first two 
problems --- Optimal Replacement (ORP) and Fault-Tolerance Score (FTS) --- have been studied in the facility-location literature, and initiate a new research direction. We also formulate and solve a fault-tolerant hitting set problem in one dimension, which may be of independent interest. 

\section{Fault-Tolerant Committees in One Dimension}
\label{sec:1d}

Even in one dimension, computing the fault-tolerance score of a given committee or finding a committee with minimum 
fault-tolerance score is nontrivial. The optimal replacement problem, however, is easy --- a simple greedy algorithm works. 
Our main result in this section is to design efficient dynamic-programming algorithms for the former two problems. 
In doing so, we also solve the fault-tolerant version of Hitting Set for points and unit segments. %, which may be of independent interest.

\subsection{Optimal Replacement Problem}
In the Optimal Replacement Problem (ORP), we are given a committee $T \subseteq C$ and a failing set $J \subseteq C$,
and we must find a replacement set $R$ minimizing the score $\sigma((T\setminus J) \cup R)$, where $|R| \leq |T \cap J|$.
Since this score is always the distance between some voter-candidate pair, it suffices to solve the following decision problem:
Is there a replacement set with score at most $r$? 
We can then try all possible $O(nm)$ distances to find the smallest feasible
replacement score.

This decision problem is equivalent to the following hitting set problem: for each voter $v \in V$, let $I_v$ be the
interval of length $2r$ centered at $v$, and let $\mathcal{I}  = \{I_v : v \in V\}$ be the set of these $n$ (voter) intervals.
A subset of candidates is a hitting set for $\mathcal{I}$ if each interval contains at least one of the candidates.
In our problem, we are given a hitting set $T$ and a failing subset of candidates $J$, and we must find the minimum-size 
replacement hitting set. Such a replacement is easily found using the standard greedy algorithm, as follows.
We first remove all of the intervals from  $\mathcal{I}$ that are already hit by a candidate in $T \setminus J$,
and we also remove all the failing candidates $J$ from $C$.
For the leftmost remaining interval, we then choose the rightmost candidate $c$ contained in it, add it to $R$,
delete all intervals hit by $c$, and iterate until all remaining intervals are hit. If we ever encounter an interval containing 
no candidate, or if the size of the replacement set is larger than $|T \setminus J|$, the answer to the decision problem is no.
Otherwise, the solution is $R$. The greedy algorithm is easily implemented to run in time $\mathcal{O}((m+n)\log (m+n))$.
To find the optimal replacement set, we can do a binary search over $O(nm)$ values of $r$ and find the smallest $r$ for which $|(T\setminus J) \cup R| \leq k$.

\begin{theorem}
    \label{thm:1d-ORP}
    The Optimal Replacement Problem can be solved in time $\mathcal{O}((m+n)\log^2(m+n))$ for one-dimensional Euclidean elections. % near-linear time algorithm
\end{theorem}

\subsection{Computing the Fault-Tolerance Score (FTS) of a Committee}
We now come to the more difficult problem of computing the fault-tolerance score $\sigma_f(T)$ of a committee $T$ in 
one dimension, which is the worst case over all possible failing sets of $T$.  Once again it suffices to solve 
the following decision problem: given a size-$k$ committee $T$ and a real number $r$, can we find a replacement with score at most $r$ for \emph{every} failing subset of size $f$?  
Using our hitting set formulation, $\sigma_f (T) \leq r$ if and only if
$T$ is an \emph{$f$-tolerant hitting set} of $\mathcal{I}$, that is, for any failing set $J \subseteq C$ of 
size at most $f$, there exists a replacement set $R \subseteq C\setminus J$ such that $|(T\setminus J) \cup R| \leq |T|$ 
and $(T\setminus J) \cup R$ hits $\mathcal{I}$. (Recall that each member of $\mathcal{I}$ is an interval of length $2r$ 
centered at one of the voter positions.)
We can then compute the fault-tolerance score of $T$ by trying each of the $O(nm)$ voter-candidate distances to 
find the smallest $r$ for which this decision problem has a positive answer.

We solve this fault-tolerant hitting set decision problem by observing that the size of a \emph{smallest} hitting set
equals the size of a \emph{maximum} independent set, defined with respect to candidate points and voter intervals
in the following way. Suppose the intervals of $\mathcal{I} = \{I_1,\dots,I_n\}$ are sorted left to right. First, we can assume without loss of generality that $|I_i \cap C| > f$ for all $i \in [n]$, since 
otherwise there is no $f$-tolerant hitting set for $\mathcal{I}$.
Given a set of points $X$ in $\mathbb{R}$, we say that a set of intervals is \textit{$X$-disjoint} if each point in $X$ is contained in at most one interval.  (That is, $X$-disjoint intervals can be thought of as \emph{independent} 
in that they contain disjoint sets of points in $X$).  The following claim is easy to prove.

\begin{lemma}
\label{obs:non-FT-structure}
Given a set of points $X$ and a set of intervals $\mathcal{J}$ on the real line, the size of a minimum hitting 
set $X' \subseteq X$ of $\mathcal{J}$ equals the maximum size of an $X$-disjoint subset of $\mathcal{J}$. 
\end{lemma}

\begin{figure}
    \centering
    \includegraphics[scale = 0.45]{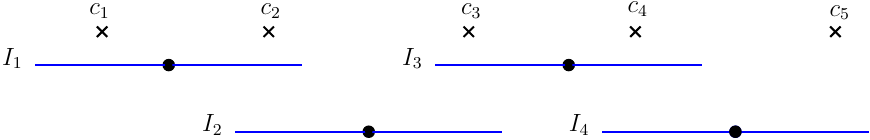}
    \caption{The figure shows an interval hitting set instance with four intervals and five points. The set $\{c_2, c_4\}$ is a feasible hitting set. For $X = \{c_2, c_3, c_5\}$, the intervals $I_1, I_3, I_4$ are $X$-disjoint.}
    \label{fig:x-disjoint-example}
\end{figure}

Thus, if $T \subseteq C$ is an $f$-tolerant hitting set for $\mathcal{I}$, then for any failing set $J \subseteq C$, 
the size of any $(C \setminus J)$-disjoint subset of $\mathcal{I}$ is at most $|T|$.  One should note that the size of the maximum  
$(C \setminus J)$-disjoint subset in $\mathcal{I}$ is a monotonically increasing function of $|J|$ --- as more candidates 
fail, more intervals can become disjoint.  Our goal is to find the maximum size of such a disjoint interval family over 
all possible failure sets $J$ of size at most $ f$.  We will do this using dynamic programming, by combining solutions 
of subproblems, where each subproblem corresponds to an index range $[i,j]$, over the set of candidate points $c_1, \ldots, c_m$.
Assuming that the candidate points $C = \{c_1,\dots,c_m\}$ are ordered from left to right,
our subproblems are defined as follows, for $1 \leq i \leq j \leq m$:
\begin{itemize}
\item 	$C_{i,j} = \{c_i,\dots,c_j\}$ is the set of candidates in the range $[c_i, c_j]$.
\item   $\mathcal{I}_{i,j} = \{I \in \mathcal{I}: I \cap C \subseteq C_{i,j}\}$ is the set of intervals that only
	contain points from $C_{i,j}$.
\item	For any $J \subseteq C_{i,j}$, $\delta_{i,j}(J)$ is the maximum size of a $(C_{i,j} \backslash J)$-disjoint 
		subset of $\mathcal{I}_{i,j}$.
\item	The subproblems we want to solve are the values $\delta_{i,j}(f) = \max_{J \subseteq C_{i,j}, |J| \leq f} \delta_{i,j}(J)$.
\end{itemize}

The key technical lemma of this section is the following claim.
\begin{lemma} \label{lem-char}
$T \subseteq C$ is an $f$-tolerant hitting set of $\mathcal{I}$ if and only if $|T \cap C_{i,j}| \geq \delta_{i,j}(f)$,
	for all $1 \leq i \leq j \leq m$.
\end{lemma}
\begin{proof}
We first show the ``if'' part of the lemma.
Assume $|T \cap C_{i,j}| \geq \delta_{i,j}(f)$ for all $i,j \in [m]$ with $i \leq j$.
To see that $T$ is an $f$-tolerant hitting set of $\mathcal{I}$, consider a failing set $J \subseteq C$ of size at most $f$.
We have to show the existence of a replacement set $R \subseteq C \backslash J$ such that $|(T \backslash J) \cup R| \leq |T|$ and $(T \backslash J) \cup R$ is a hitting set of $\mathcal{I}$.
We write $T \backslash J = \{c_{i_1},\dots,c_{i_p}\}$, where $i_1 < \cdots < i_p$.
For convenience, set $i_0 = 0$ and $i_{p+1} = m+1$.
By our assumption, every interval $I \in \mathcal{I}$ is hit by some point in $C$.
Thus, either $I$ is hit by $T \backslash J$ or $I$ belongs to $\mathcal{I}_{i,j}$ where $i = i_{t-1}+1$ and $j = i_t-1$ for some index $t \in [p+1]$.
Now consider an index $t \in [p+1]$.
%Let $i = i_{t-1}+1$ and $j = i_t-1$.
We write $T_t = T \cap C_{i,j}$ and define $R_t \subseteq C_{i,j} \backslash J$ as a minimum hitting set of $\mathcal{I}_{i,j}$.
By Lemma~2, the size of $R_t$ is equal to the maximum size of a $(C_{i,j} \backslash J)$-disjoint subset of $\mathcal{I}_{i,j}$, which is nothing but $\delta_{i,j}(J \cap C_{i,j})$.
Also, by assumption, we have $|T_t| = |T \cap C_{i,j}| \geq \delta_{i,j}(f) \geq \delta_{i,j}(J \cap C_{i,j})$.
Therefore, $|R_t| \leq |T_t|$.
Finally, we define $R = \bigcup_{t=1}^{p+1} R_t$.
Clearly, $(T \backslash J) \cup R$ hits $\mathcal{I}$.
So it suffices to show that $|(T \backslash J) \cup R| \leq |T|$.
Since $|R_t| \leq |T_t|$ for all $t \in [p+1]$, we have
\begin{equation*}
    |(T \backslash J) \cup R| = |T \backslash J| + \sum_{t=1}^{p+1} |R_t| \leq |T \backslash J| + \sum_{t=1}^{p+1} |T_t| = |T|,
\end{equation*}
which completes the proof of the ``if'' part.

Next, we prove the ``only if'' part of the lemma.
Assume $T \subseteq C$ is an $f$-tolerant hitting set of $\mathcal{I}$.
Consider two indices $i,j \in [m]$ with $i \leq j$.
To show $|T \cap C_{i,j}| \geq \delta_{i,j}(f)$, it suffices to show that $|T \cap C_{i,j}| \geq \delta_{i,j}(J)$ for all $J \subseteq C_{i,j}$ with $|J| \leq f$.
Since $T$ is an $f$-tolerant hitting set of $\mathcal{I}$, there exists $R \subseteq C \backslash J$ such that $|(T \backslash J) \cup R| \leq |T|$ and $(T \backslash J) \cup R$ is a hitting set of $\mathcal{I}$.
For brevity, let $T' = (T \backslash J) \cup R$.
By definition, the intervals in $\mathcal{I}_{i,j}$ can only be hit by the points in $C_{i,j}$.
Thus, $T' \cap C_{i,j}$ is a hitting set of $\mathcal{I}_{i,j}$.
As $T' \cap C_{i,j} \subseteq C_{i,j} \backslash J$, by Lemma~\ref{obs:non-FT-structure}, the size of $T' \cap C_{i,j}$ is at least the maximum size of a $(C_{i,j} \backslash J)$-disjoint subset of $\mathcal{I}_{i,j}$, i.e., $|T' \cap C_{i,j}| \geq \delta_{i,j}(J)$.
Furthermore, because $J \subseteq C_{i,j}$, we have $(T \backslash J) \backslash C_{i,j} = T \backslash C_{i,j}$.
It follows that $T \backslash C_{i,j} \subseteq T' \backslash C_{i,j}$ and thus $|T \backslash C_{i,j}| \leq |T' \backslash C_{i,j}|$.
For a committee $T$, we can partition $T$ into two parts: the part containing candidates in $C_{i,j}$ and the part containing candidates outside of $C_{i,j}$.
Hence, $|T| = |T \cap C_{i,j}| + |T \backslash C_{i,j}|$ and $|T'| = |T' \cap C_{i,j}| + |T' \backslash C_{i,j}|$.
Because $|T'| \leq |T|$ and $|T' \backslash C_{i,j}| \geq |T \backslash C_{i,j}|$, we have $|T' \cap C_{i,j}| \leq |T \cap C_{i,j}|$.
Therefore, $|T \cap C_{i,j}| \geq \delta_{i,j}(J)$.
%Assume $|T \cap C_{i,j}| \geq \delta_{i,j}(f)$ for all $i,j \in [m]$ with $i \leq j$.
This completes the proof of Lemma~\ref{lem-char}.
\end{proof}

In order to decide if $\sigma_f(T) \leq r$,  therefore, we just have to compute $\delta_{i,j}(f)$, for all $i,j$, 
and check the condition $|T \cap C_{i,j}| \geq \delta_{i,j}(f)$. We now show how to do that efficiently.

\paragraph*{Efficiently Computing $\delta_{i,j} (f)$}
\label{subsubsec:delta_ij}
For ease of presentation, we show how to compute $\delta_{1,m}(f)$; computing other $\delta_{i,j}(f)$ is similar.
We have $C_{1,m} = C$, $\mathcal{I}_{1,m} = \mathcal{I}$, and $\delta_{1,m}(f)$ is size of the largest subset 
of $\mathcal{I}$ that is $(C \backslash J)$-disjoint for any failing set $J \subseteq C$ with $|J| \leq f$.
The intervals of $\mathcal{I} = \{I_1, I_2, \ldots, I_n\}$  are in the left to right sorted order and, for each $i \in [n]$, 
let $C(I_i) = C \cap I_i$ be the set of points in $C$ that hits $I_i$.
Define $\varGamma[i][j]$ as the maximum size of an $(C\setminus X)$-disjoint subset $\mathcal{J} \subseteq \{I_1,\dots,I_i\}$ such that 
$X \subseteq C$ and $|X| \leq j$.

\begin{lemma}
We have the following recurrence
    $$ \varGamma[i][j] ~=~ \max \left\{ \begin{array}{l}
                 \varGamma[i-1][j], \\
                \max\limits_{0 \leq i' \leq i} 1 + \varGamma[i'][j - |C(I_i) \cap C(I_{i'})|]
            \end{array}\right\}
$$ 
\end{lemma}

Clearly, $\delta_{1,m}(f) = \varGamma[n][f]$.
The base case for our dynamic program is $\varGamma[0,j] = 0$ for all $j \in [f]$ and $\varGamma[i][j] = -\infty$ for $j < 0$ and all $i \in [n]$.
Our dynamic program runs in time $\mathcal{O}(n^2mf)$.
In the same way, we can compute the values of $\delta_{i,j}(f)$ for all $i,j \in [m]$ with $i\leq j$.

\begin{lemma}\label{lem-delta}
    $\delta_{i,j}(f)$, for all $1 \leq i \leq j \leq m$,  can be computed in time $\mathcal{O}(n^2m^3f)$.
\end{lemma}

Given a hitting set $T \subseteq C$ and the values $\delta_{i,j}(f)$, we can verify the condition in Lemma~\ref{lem-char} in 
time $\mathcal{O}(m^3)$.  We can then use binary search to find the smallest value of $r$ for which $T$ is an $f$-tolerant 
hitting set. This establishes the following result.

\begin{theorem}
    \label{thm:1d-FTS}
    The fault-tolerance score of a $1$-dimensional committee $T$ can be computed in time $\mathcal{O}(n^2m^3f \log(nm))$.
\end{theorem}

\subsection{Optimal Fault-Tolerant Committee}
\label{sec:OFTC-1d}
We now address the problem of \emph{designing} a fault-tolerant committee: select a committee $T$ of size $k$ 
whose fault-tolerance score $\sigma_f (T)$ is minimized. Thus, our goal is \emph{not} to optimize the fault-free 
score of $T$, namely $\sigma_0 (T)$, but rather the score that the best replacement
will have after a worst-case set of $f$ faults in $T$, namely $\sigma_f (T)$.
Following the earlier approach, we again focus on the decision question: given some $r \geq 0$, is there a committee 
of size $k$ with $\sigma_f(T) \leq r$? For a given value of $r$, we construct our hitting set instance with 
candidate-points and voter-intervals, and compute a minimum-sized $f$-tolerant hitting set $T \subseteq C$ as follows:

\begin{enumerate}
\item Compute the value of $\delta_{i,j}(f)$, for all $1 \leq i \leq j \leq m$.
\item Compute a minimum subset $T \subseteq C$ satisfying $|T \cap C_{i,j}| \geq \delta_{i,j}(f)$, for all $1 \leq i \leq j \leq m$.
\item If $|T| \leq k$, we have a solution; otherwise, the answer to the decision problem is no.
\end{enumerate}

Step (1) is implemented using the dynamic program of the previous subsection, and so it suffices to explain
how to implement step (2).
We assume without loss of generality that $|C_{i,j}| \geq \delta_{i,j}(f)$ for all $i,j$, because otherwise there is no solution.
We compute a set $T$ using the following greedy algorithm.
\begin{itemize}
\item	Initialize $T = \emptyset$.
\item	For each $c_k$ for $k\in [m]$, if there exists $i,j \in [m]$ with $i \leq k \leq j \leq m$ such that 
	$\delta_{i,j}(f) \geq |T \cap C_{i,j}| + (j-k+1)$, then add $c_k$ to $T$.%; otherwise skip $c_k$.
\end{itemize}

The algorithm runs in time $\mathcal{O}(m^3)$. 
To prove correctness, we first claim the following.

\begin{lemma} \label{obs-valid}
$|T \cap C_{i,j}| \geq \delta_{i,j}(f)$, for all $1 \leq i \leq j \leq m$.
\end{lemma}
\begin{proof}
Suppose not, so we have $|T \cap C_{i,j}| < \delta_{i,j}(f)$, for some $i \leq j$.
We recall that for any interval $I_i \in \mathcal{I}$, $|I_i \cap C| > f$.
Therefore, for any failing set $J$, $C_{i,j} \setminus J$ is a hitting set of $\mathcal{I}_{i,j}$, and $|C_{i,j}| \geq \delta_{i,j}(f)$.
This implies that there exists some point among $c_i,\dots,c_j$ that is not in $T$.
Let $k \in \{i,\dots,j\}$ be the largest index such that $c_k \notin T$.
For convenience, we use $T'$ to denote the set $T$ in the iteration of our algorithm that considers $c_k$.
Note that $T \cap C_{i,j} = (T' \cap C_{i,j}) \cup \{c_{k+1},\dots,c_j\}$ and $(T' \cap C_{i,j}) \cap \{c_{k+1},\dots,c_j\} = \emptyset$.
Therefore, $|T' \cap C_{i,j}| = |T \cap C_{i,j}| - (j-k) < \delta_{i,j}(f) - (j-k)$.
This implies $|T' \cap C_{i,j}|+(j-k+1) \leq \delta_{i,j}(f)$.
By our algorithm, in this case we should include $c_k$ in $T$, which contradicts the fact that $c_k \notin T$.
%This contradiction implies that the observation is true.
\end{proof}

We now argue that $T$ has the minimum size among all subsets of $C$ satisfying the property of Lemma~\ref{obs-valid}.
Let $\mathsf{opt}$ be the minimum size of a subset of $C$ satisfying the desired property.
We write $T = \{c_{k_1},\dots,c_{k_r}\}$, where $k_1 < \cdots < k_r$.

\begin{restatable}[]{lemma}{Greedyproof}
\label{lem:1d-greedy-proof}
For any $t \in [r]$, there exists a subset $T^* \subseteq C$ such that 
	(1) $|T^* \cap C_{i,j}| \geq \delta_{i,j}(f)$ for all $i,j \in [m]$ with $i \leq j$,
    	(2) $|T^*| = \mathsf{opt}$, and
    	(3) $\{c_{k_1},\dots,c_{k_t}\} \subseteq T^*$.
\end{restatable}
\begin{proof}
We prove the observation by induction on $t$.
For $t=0$, the statement trivially holds.
Suppose the statement holds for $t-1$, i.e., there exists a subset $T^* \subseteq C$ satisfying the first two conditions in Lemma~\ref{lem:1d-greedy-proof} and $\{c_{k_1},\dots,c_{k_{t-1}}\} \subseteq T^*$.
We show the statement holds for $t$.
Specifically, we shall modify $T^*$ to make it satisfy $\{c_{k_1},\dots,c_{k_t}\} \subseteq T^*$ while maintaining the first two conditions in the observation.

First, we notice that $T^*$ must contain a point other than $c_{k_1},\dots,c_{k_{t-1}}$.
To see this, suppose $T^* = \{c_{k_1},\dots,c_{k_{t-1}}\}$.
Since our algorithm added $c_{k_t}$ to $T$, there exist $i,j \in [m]$ with $i \leq k_t \leq j$ such that $\delta_{i,j}(f) \geq |\{c_{k_1},\dots,c_{k_{t-1}}\} \cap C_{i,j}| + (j-k_t+1)$.
This implies that $\delta_{i,j}(f) > |\{c_{k_1},\dots,c_{k_{t-1}}\} \cap C_{i,j}|$, that is, $\delta_{i,j}(f) > |T^* \cap C_{i,j}|$, which contradicts the fact that $T^*$ satisfies the first condition in the lemma.
Thus, $T^*$ contains a point other than $c_{k_1},\dots,c_{k_{t-1}}$.

Now let $k$ be the smallest index such that $c_k \in T^* \backslash \{c_{k_1},\dots,c_{k_{t-1}}\}$.
If $k=k_t$, then $\{c_{k_1},\dots,c_{k_t}\} \subseteq T^*$ and we are done.
Otherwise, we remove $c_k$ from $T^*$ and add $c_{k_t}$ to $T^*$.
After this modification, it is clear that $|T^*| = \mathsf{opt}$ and $\{c_{k_1},\dots,c_{k_t}\} \subseteq T^*$.
So it suffices to show $|T^* \cap C_{i,j}| \geq \delta_{i,j}(f)$ for all $i,j \in [m]$ with $i \leq j$.
Consider indices $i,j \in [m]$ with $i \leq j$.
By assumption, before the modification we have $|T^* \cap C_{i,j}| \geq \delta_{i,j}(f)$.
If $j \geq k_t$, then $|T^* \cap C_{i,j}|$ does not decrease after the modification, and is thus at least $\delta_{i,j}(f)$.
So assume $j < k_t$.
In this case, $T \cap C_{i,j} = \{c_{k_1},\dots,c_{k_{t-1}}\} \cap C_{i,j} \subseteq T^* \cap C_{i,j}$.
Since $|T \cap C_{i,j}| \geq \delta_{i,j}(f)$ by Lemma~3, we have $|T^* \cap C_{i,j}| \geq \delta_{i,j}(f)$.
\end{proof}

We use binary search to find the smallest $r$ such that the reduced instance has an $f$-tolerant hitting set of size at most $k$.
Therefore, the following theorem holds.

\begin{theorem}
    \label{thm:1d-OFTC}
    Optimal Fault-Tolerant Committee can be solved in time $\mathcal{O}(n^2m^3f\log(nm))$ for one-dimensional Euclidean elections.
\end{theorem}

\begin{remark}
    Our dynamic programming algorithm works as long as either the set $V$ or the set $C$ is embedded in $\mathbb{R}$ (i.e., has a linear ordering), while the other set can have an arbitrary $d$-dimensional embedding.
    Moreover, we can also extend our algorithms to ordinal elections with (widely studied) single-peaked preferences \cite{betzler2013computation,misra2019robustness} to compute an optimal fault-tolerant Chamberlin-Courant committee.
\end{remark}

\section[Fault tolerant in Higher Dimensions]{Fault-Tolerant Committees in Multidimensional Space}
\label{sec:2d-section}
We now consider fault tolerance in multidimensional elections. Unsurprisingly, the optimal committee 
design problem is intractable --- it is similar to facility location --- but it turns out that the seemingly simpler 
variants ORP and FTS are also intractable.  

\subsection{Hardness Results}

\begin{restatable}[$\star$]{theorem}{hardnesstwod}
\label{thm-hardness-2d}
All three problems (Optimal Replacement, Fault-Tolerance Score, and Optimal Fault-Tolerant Committee) are NP-hard, 
in any dimension $d \geq 2$ under the Euclidean norm, where size of the committee $k$ and the failure
parameter $f$ are part of the input.
\end{restatable}

We will use a single construction to show NP-hardness all three problems.
Our proof uses a reduction from the NP-complete problem \textsc{Planar Monotone 3-SAT} (PM-3SAT) \cite{deBerg2010SAT}.
An input to this problem is a \emph{monotone} 3-CNF formula $\phi$ where each clause contains either three positive literals or three negative literals, and whose variable-clause incidence graph has a planar embedding which is given as a part of the input. 
Given an instance $\phi$ of PM-3SAT, our reduction  constructs a 2-dimensional Euclidean election.
The general outline follows a scheme used in~\cite{sonar2022} to show the hardness of committee selection under \emph{ordinal} preferences, but generalizing the proof to \emph{fault-tolerant committees} requires several technical modifications and a new proof of correctness.

\begin{figure*}[t!]
    \centering
    \begin{subfigure}[b]{0.5\textwidth}
        \centering
        \includegraphics[scale=0.35]{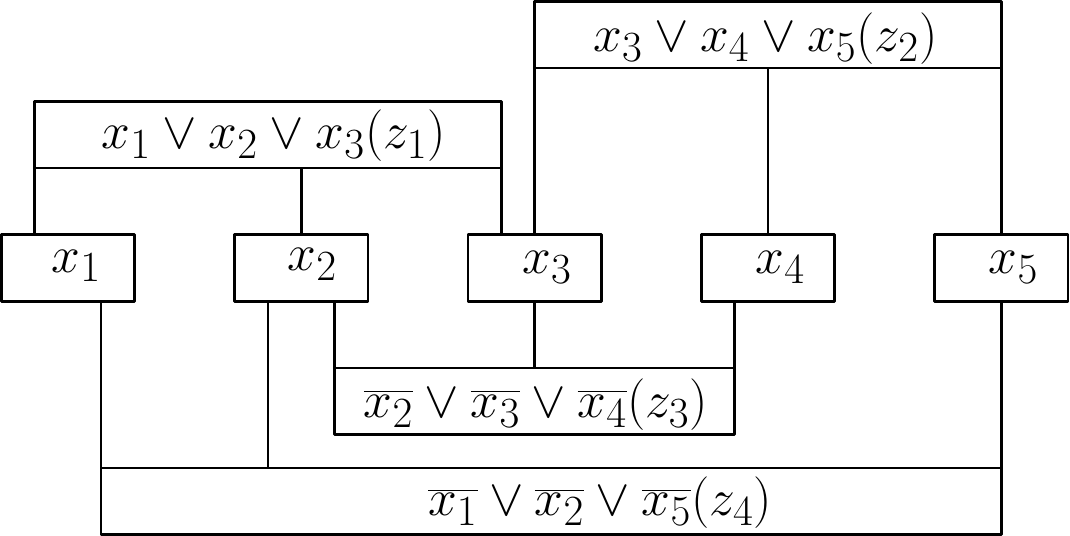}
        \caption{}
        \label{fig:pm-3sat}
    \end{subfigure}%
    \begin{subfigure}[b]{0.5\textwidth}
        \centering
        \includegraphics[scale=0.35]{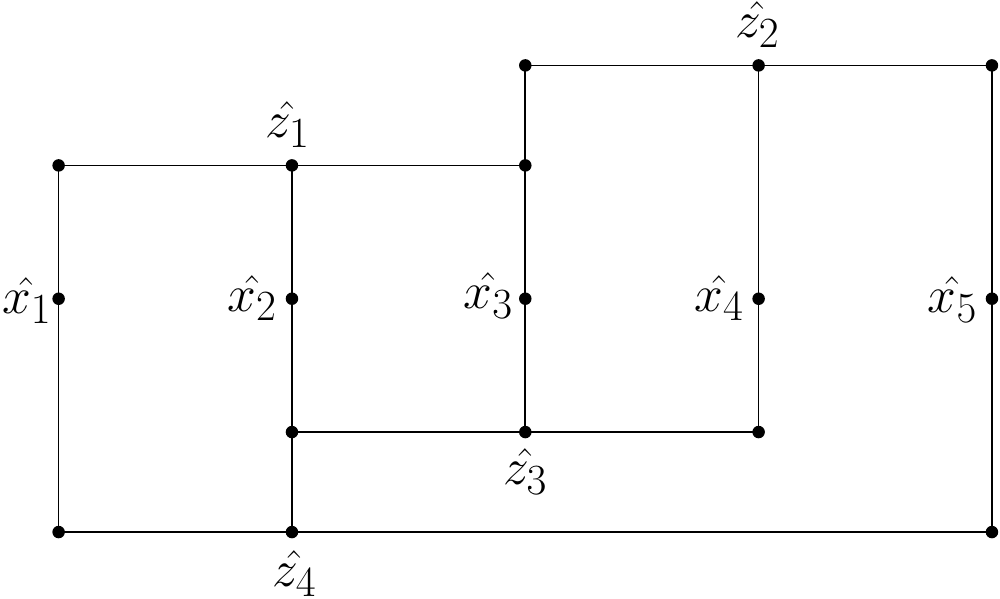}
        \caption{}
        \label{fig:orthogonal-embedding}
    \end{subfigure}
    \caption{Figure (a) shows rectangular embedding of the PM-3SAT instance given as a part of the input. In figure (b), we show a transformation of rectangular embedding to orthogonal embedding useful for our construction.
    Here, each variable $x_i$ in the rectangular embedding is replaced by the variable reference point $\hat{x_i}$ and each clause is replaced $z_i$ is replaced by the clause reference point $\hat{z_i}$.}
\end{figure*}

For ease of referencing, we adapt the terminologies from \cite{sonar2022}.
In the planar embedding of the formula $\phi$, each variable/clause is drawn as an (axis-parallel) rectangle in the plane, and so this is called a \textit{rectangular embedding}. 
See Figure~\ref{fig:pm-3sat} for an illustration.
The rectangles for the variables are drawn along the $x$-axis, while the rectangles for the positive (resp., negative) clauses lie above (resp., below) the $x$-axis.
If a clause contains a variable, then there is a vertical segment connecting the clause rectangle and the variable rectangle.
Each such vertical segment is disjoint from all the rectangles except the two it connects. 

The rectangular embedding of $\phi$ can be modified to another embedding which is easier to work with called \emph{orthogonal embedding}.
We refer the reader to \cite{sonar2022multiwinner} for details of the modification (See figure~\ref{fig:orthogonal-embedding} for intuition).
The intersection points of vertical and horizontal segments in the orthogonal embedding are \emph{connection points}.
To build the intuition for the orthogonal embedding, we now give its properties as stated in \cite{sonar2022multiwinner}:

\begin{enumerate}
    \item [(i)] Vertical and horizontal segments do not cross.
    \item [(ii)] Each horizontal segment corresponds to a clause in $\phi$.
    Moreover, it intersects exactly three vertical segments corresponding to the literals in that clause.
    \item [(iii)] The endpoints of all segments are connection points.
\end{enumerate}

For each horizontal segment, we will refer to the middle connection point as the \emph{reference point} of a the clause (notice that from properties (ii) and (iii) each horizontal segment has three connection points and two of those are the left and the right endpoint of the segment).
The reference points for the variables $x_1, \ldots, x_n$ are denoted by $\hat{x_1}, \ldots, \hat{x_n}$ and for the clauses $z_1, \ldots, z_m$ they are denoted by $\hat{z_1}, \ldots, \hat{z_m}$. 
With shifting/scaling of the orthogonal embedding without changing its topology, we can ensure that the $x$-coordinates ($y$-coordinates) of vertical (resp., horizontal) segments are distinct even integers in the range $\{1, \ldots, 2n\}$ (resp., $\{-2m, \ldots, 2m\}$).
This guarantees that all connection points have even integer coordinates and the embedding is contained in $[1,2n] \times [-2m, 2m]$ rectangle.
Now using the integral points on each segment $s$, we can partition $s$ into $\ell(s)$ parts each of unit length where $\ell(s)$ is the length of $s$.
These unit length segments are called \emph{pieces} of the orthogonal embedding.
We use $N$ to denote the total number of pieces.
Note that $N = \mathcal{O}(nm)$.

We now construct a Euclidean election $E = (V, C)$ with voters and candidates as points in $\mathbb{R}^2$.

\noindent \textbf{Variable gadgets.} For each variable $x_i$, we choose two additional points near (but not equal to) $x_i$ as follows.
Recall that there are two vertical pieces incident to $\hat{x}_i$ in the orthogonal embedding: one above $\hat{x}_i$, and the other below $\hat{x}_i$.
We choose a point with distance $0.2$ from $\hat{x}_i$, on each of the two pieces.
Next, we put $f+1$ candidates on each of these two points and a (single) candidate at $\hat{x_i}$ (we set the value of $f$ later in the construction).
Furthermore, we place a voter on each of these three points.
We call these candidates/voters the $x_i$-\textit{gadget}.
See figure~\ref{fig:vertex-gadget}.
For $i\in [n]$, we construct the $x_i$-gadget for each variable $x_i$.
Overall, the variable gadgets have $(2f+3)n$ candidates and $3n$ voters.

\noindent \textbf{Clause gadgets.} Next, we construct a set of candidates/voters for the clauses $z_1,\dots,z_m$.
For each clause $z_i$, we put a voter at the reference point $\hat{z}_i$, and call this voter the $z_i$-\textit{gadget}.
See figure~\ref{fig:edge-gadget}.
The total number of voters in the clause gadgets is $m$.
Clause gadgets do not have any candidates.

\noindent \textbf{Piece gadgets.} Finally, we construct a set candidates/voters to connect the variable and clause gadgets.
Consider a piece $s$ of the orthogonal embedding.
Recall that $s$ is a unit-length segment.
Let $s^-$ and $s^+$ be the two endpoints of $s$.
We identify these endpoints as follows:
For a vertical piece $s$ above (resp., below) the $x$-axis, we say $s^-$ is the bottom (resp., top) endpoint of $s$ and $s^+$ is the top (resp., bottom) endpoint of $s$.
For a horizontal piece $s$, $s$ must belong to the horizontal segment of some clause $z_i$.
Suppose $s$ is to the left (resp., right) of the reference point $\hat{z}_i$, then $s^-$ is the left (resp., right) endpoint of $s$ and $s^+$ be the right (resp., left) endpoint of $s$.
For every piece $s$ that is \textit{not} adjacent to any clause reference point, we choose four points near $s^+$ and add candidates/voters on them as follows: 
We place $f+1$ candidates each on a point which is $0.3$ below and $0.3$ to the right of $s^+$, and on a point which is $0.4$ above and $0.3$ to the left of $s^+$, and on a point at $s^+$.
Further, we place a candidate at a point which is $0.3$ above $s^+$.
Lastly, we place one voter at each of these four points.
We call these the candidates/voters of the $s$-\textit{gadget}.
See figure~\ref{fig:piece-gadget}.
Note that pieces adjacent some clause reference points do not have gadgets.
Therefore, the total number of candidates in the piece gadgets is $(3f+4)(N-3m)$, as each clause reference point is adjacent to three pieces, and the number of voters is $4 (N-3m)$.

Combining these three constructed gadgets, our election $E=(V, C)$ has $4N - 11m+3n$ voters and $(3N-9m+2n)f + 4N - 12m + 3n$ candidates. 
We set the committee size $k$ equals to $N+n-3m$.
%, the fault-tolerance parameter $f$ equal to $k$, and the committee score $r$ equal to $0.75$.
Clearly, the construction can be done in a polynomial time.
The main intuition behind the construction is the following.

Candidates in the constructed instance can be partitioned into two types:

\begin{itemize}
    \item \emph{Robust candidates} $(C_{rob} \subseteq C)$ is the set of candidates such that for each candidate, there are $f$ other candidates at the exact same location as it.
    Note that for any failing set $J \subseteq C$, at least one candidate is live in each of these locations.
    \item \emph{Covering candidates} $(C_{cov} \subseteq C)$ is the set of candidates such that for each candidate in the set, it is the unique candidate at its location.
    Note that $|C_{cov}| = k$.
\end{itemize}

\begin{figure*}[t!]
    \centering
    \begin{subfigure}[b]{0.4\textwidth}
        \centering
        \includegraphics[height = 3cm, width = 6cm]{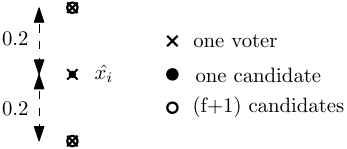}
        \caption{Vertex gadget}
        \label{fig:vertex-gadget}
    \end{subfigure}%
    \begin{subfigure}[b]{0.3\textwidth}
        \centering
        \includegraphics[scale=1]{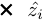}
        \caption{Edge gadget}
        \label{fig:edge-gadget}
    \end{subfigure}%
    \begin{subfigure}[b]{0.3\textwidth}
        \centering
        \includegraphics[scale=1.2]{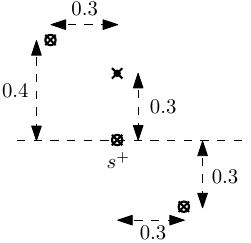}
        \caption{Piece gadget}
        \label{fig:piece-gadget}
    \end{subfigure}
    \caption{Gadgets in our construction.
    Here, a disk and a cross at the same location indicates a voter and a candidate at the same location. Similarly, a circle and a cross at the same location indicates a voter and $(f+1)$ candidates at the same location.}
\end{figure*}

In the constructed election, the following lemma holds.
\begin{lemma}
\label{lem:distance-to-T}
    In the constructed election $E = (V, C)$, we have $\sigma(C_{cov}) \leq 0.75$ where $C_{cov} \subseteq C$ is the set of covering candidates.
\end{lemma}
\begin{proof}
For all voters $v \in V$, we will show that $d(v, C_{cov}) \leq 0.75$.

First, we consider the voters in the variable gadget. 
For a variable $x_i$, the candidate at $\hat{x}_i$ belongs to $C_{cov}$.
All three voters in the variable gadget are within distance $0.2$ from $\hat{x}_i$.
Therefore, all the voters in the variable gadget have $d(v, C_{cov}) \leq 0.75$.

Next, we consider the clause-gadgets.
The closest candidate in $T$ for a voter placed at $\hat{z}_i$ is at a distance $0.7$ from it; therefore, all voters in the clause gadgets all have $d(v, C_{cov}) \leq 0.75$.

Finally, we consider the piece gadgets.
For each piece $s$, $C_{cov}$ contains a candidate at a distance $0.3$ above $s^+$.
It can be verified that all four voters in the piece gadget have their closest candidate in $C_{cov}$ within distance $\sqrt{0.45} < 0.7$ (See figure~\ref{fig:piece-gadget}).

Hence, for all voters in the piece gadgets, $d(v,C_{cov}) \leq 0.75$.
This completes the proof of Lemma~\ref{lem:distance-to-T}.
\end{proof}

The election constructed above will be used to show hardness for all three problems.
Let the distance $r = 0.75$.
Recall that $k = N + n - 3m$.
We will now, we describe the decision version of each of our problems along with the construction of additional elements necessary in their input:

\begin{enumerate}
    \item [(i)] \textbf{ORP}: In the input of ORP we additionally need a committee $T \subseteq C$ and a failing set $J \subseteq C$.
    We set $T = J = C_{cov}$.
    
    We ask if there exists a replacement $R \subseteq C\setminus C_{cov}$ such that $\sigma(R) \leq r$?
    
    \item [(ii)] \textbf{FTS}: In the input of FTS we additionally need a committee $T \subseteq C$ and a fault-tolerance parameter $f$.
    We set $T = C_{cov}$ and $f = k$.
    
    We ask if $\sigma_f(T) \leq r$?
    
    \item [(iii)] \textbf{OFTC}: In the input of OFTC we additionally need a committee size $k'$ and a fault-tolerance parameter $f$.
    We set $k' = k$ and $f = k$.
    
    We ask if there exists a $k$-sized committee $T \subseteq C$ with $\sigma_f(T) \leq r$?
\end{enumerate}

To show the equivalence, we will show that the answer to each of the above question is yes if and only if $\phi$ is satisfiable.

The following lemma shows the proof of equivalence for problem (i).

\begin{lemma}
    \label{lem:ORP-equivalence}
    There exists a committee $R \subseteq C\setminus C_{cov}$ with size $k$ such that $\sigma(R) \leq r$ if and only if $\phi$ is satisfiable where $k = |C_{cov}| = N + n - 3m$, $N$ is the total number of pieces, and $m$ (resp., $n$) is the number of clauses (resp., variables) in $\phi$, respectively.
\end{lemma}

In the next two subsections, we give the proof of Lemma~\ref{lem:ORP-equivalence}.

\subsubsection[]{Forward direction of Lemma~\ref{lem:ORP-equivalence}}
\label{subsec:ORP-equivalence}
In this section, we will show that if $\phi$ is satisfiable, then there exists a $k$-sized committee $R \subseteq C\setminus C_{cov}$ such that $\sigma(R) \leq r$ (recall that $r= 0.75$).

Suppose $\phi$ is satisfiable. 
Let $\pi: X \rightarrow \{\mathsf{true},\mathsf{false}\}$ be an assignment which makes $\phi$ true.
We construct a $k$-sized committee $R \subseteq C\backslash C_{cov}$ with $\sigma(R) \leq r$ using $\phi$.
We include one candidate from the every variable gadget and every piece gadget to $R$ as follows:

\begin{itemize}
    \item \textit{Replacement candidates from variable gadgets:}
    Consider a variable $x_i$.
    By our construction, the $x_i$-gadget contains $2f+3$ candidates which have the same $x$-coordinates as $\hat{x}_i$.
    If $\pi(x_i) = \mathsf{true}$ (resp., $\pi(x_i) = \mathsf{false}$), we include in $R$ one of the topmost (resp., bottommost) candidates in the $x_i$-gadget.
    
    \item \textit{Replacement candidates from piece gadgets:}
    Consider a piece $s$ not adjacent to a clause reference point (recall that pieces adjacent to some clause reference point do not have gadgets on them).
    We begin by defining a variable as an \emph{associated} variable of $s$ in the same way as described in \cite{sonar2022multiwinner}:
    When $s$ is vertical, the associated variable of $s$ is just the variable whose vertical segment contains $s$.
    For when $s$ is horizontal, then $s$ must belong to the horizontal segment of some clause $z_j$.
    Then, if $s$ to the left of the reference point $\hat{z_j}$ then the associated variable of $s$ is the variable whose vertical segment intersects the left endpoint of the horizontal segment of $z_j$, and vice versa when $s$ to the right of the reference point $\hat{z_j}$.
    
    Let $x_i$ be the associated variable of $s$. Then,
    \begin{enumerate}
        \item [(i)] If $\pi(x_i) = \mathsf{true}$: We include in $R$ a candidate in the $s$-gadget that is $0.4$ above and $0.3$ to the left (resp., $0.4$ below and $0.3$ to the right) of $s^+$ if $s$ is above (resp., below) the $x$-axis.
        \item [(ii)] If $\pi(x_i) = \mathsf{false}$: We include in $R$ a candidate in the $s$-gadget that is $0.3$ below and $0.3$ to the right (resp., $0.3$ above and $0.3$ to the left) of $s^+$ if $s$ is below (resp., above) the $x$-axis.
    \end{enumerate}
\end{itemize}

This finishes the construction of the committee $R$.
Recall that the total number of variable and the piece gadgets is $N + n - 3m = k$. 
Therefore, $|R| = k$.
The following lemma completes the ``if'' part of Lemma~\ref{lem:ORP-equivalence}.

\begin{lemma}
\label{lem:replacement-score-bound}
$\sigma(R) \leq r$.
\end{lemma}
\begin{proof}
For all voters $v \in V$ in the constructed instance, we will show that $d(v, R) \leq r$.

First, we consider the voters in the variable gadget. 
For a variable $x_i$, either a candidate $0.2$ above or $0.2$ below $\hat{x}_i$ belongs to $R$.
Hence, for all three voters, $d(v, R) \leq 0.6$.
Recall that $r = 0.75$.
Therefore, all the voters in the variable gadget have $d(v, T) \leq r$.

Next, we consider the clause-gadgets.
For a clause $z_i$, consider the voter $v$ placed at $\hat{z}_i$.
The closest candidate in $R$ for this voter, is the candidate placed at $0.4$ above and $0.3$ to the left of $s^+$ from a clause gadget with $d(\hat{z}_i, s^+) = 1$.
Hence, $d(v, R) = \sqrt{0.45} < r$.

Finally, we consider the piece gadgets.
We consider two cases: piece gadget not adjacent to a variable reference point and piece gadget adjacent to a variable reference point.
Let $x_{i}$ be the associated variable with the piece gadget $s$, and $\pi(x_i) = \mathsf{true}$ (the case when $\pi(x_i) = \mathsf{false}$ is symmetric).
Let $s$ be a piece gadget not adjacent to any variable reference point.
Suppose $s$ is above $x$-axis, and $\hat{s}$ is the gadget below $s$ (the case when $s$ is below the $x$-axis is symmetric; hence, we leave the proof for that to the reader).
Recall that the piece gadgets contains four voters:
Let $v_1, v_2, v_3$, and $v_4$ be the voters placed at distances $0.4$ above and $0.3$ to the left of $s^+$, $0.3$ above $s^+$, at $s^+$, and $0.3$ below and $0.3$ to the right of $s^+$, respectively.
We know that a candidate placed at the location of voter $v_1$ (say $c_1$) belongs to $R$.
Notice $v_1, v_2$ and $v_3$ have a distance at most $0.5$ from this voter but $v_4$ has a distance $\sqrt{0.85} > r$.
But consider the gadget corresponding to $\hat{s}$.
Here, we know $R$ contains a candidate (say $c_1'$) placed at $0.4$ above and $0.3$ to the left of $\hat{s}^+$.
It can be verified that $d(v_4, c_1') = \sqrt{0.45} < r$.
Therefore, $d(v_4, R) \leq r$.
We now consider the case when $s$ is adjacent to a variable reference points (say $\hat{x}_i$).
We know $s$ has four voters $v_1, v_2, v_3$, and $v_4$ placed at distances $0.4$ above and $0.3$ to the left of $s^+$, $0.3$ above $s^+$, at $s^+$, and $0.3$ below and $0.3$ to the right of $s^+$, respectively.
We know that a candidate placed at the location of voter $v_1$ (say $c_1$) belongs to $R$, and this candidate is at a distance at most $0.5$ from $v_1, v_2$, and $v_3$.
Recall that $\pi(x_i) = \mathsf{true}$, $R$ contains a candidate at a distance $0.2$ above $\hat{x_i}$.
Notice that this candidate is at a distance $\sqrt{0.34}$ from $v_4$.
Therefore, $d(v_4, R) < r$.
This completes the proof of Lemma~\ref{lem:replacement-score-bound}.
\end{proof}

\subsubsection{Reverse direction of Lemma~\ref{lem:ORP-equivalence}}

Suppose $R \subseteq C\setminus C_{cov}$ is a $k$-sized committee with $\sigma(R) \leq r$.
We will show how to recover a satisfying assignment $\pi: X \rightarrow \{\mathsf{true},\mathsf{false}\}$ using $R$. 
The structure of our proof is similar to the reverse direction argument in \cite{sonar2022multiwinner}.
First, we observe the following property of $R$.

\begin{lemma} \label{lem-onecand}
$R$ contains exactly one candidate from every variable and piece gadget. 
\end{lemma}
\begin{proof}
We begin with the variable gadgets.
Consider an $x_i$-gadget corresponding to the variable $x_i$. 
Recall that we place three voters in the $x_i$-gadget: One voter at $\hat{x_i}$ (say $v_1$), and one voter each at a distance $0.2$ above and $0.2$ below of $\hat{x_i}$ .
Observe that for $v_1$, its distance to any candidate from the adjacent piece gadgets is at least $\sqrt{(0.7)^2 + (0.3)^2} = \sqrt{0.58}$  which is strictly greater than $r$.
Hence, $R$ must include at least one candidate from the $x_i$-gadget to ensure $d(v_1, R) \leq r$.
We now consider the piece gadgets.
Let $v_1$ be the voter placed at $s^+$.
The nearest candidate to $v_1$ from the adjacent piece gadgets is at a distance at least $\sqrt{0.58}$ which is strictly greater than $r$.
Therefore, $R$ must contain at least one candidate from each piece gadget to ensure $d(v_1, R) \leq r$.

Finally, recall that the committee size is $k = N+n-3m$ and the number of variable (piece) gadgets is $n (N-3m)$, respectively.
Therefore, by a simple counting argument, we can conclude that $T$ must contain exactly one candidate from each variable gadget and each piece gadget.
This completes the proof of Lemma~\ref{lem-onecand}.
\end{proof}

We will now use $R$ to recover a satisfying assignment $\pi$ for $\phi$.
For an arbitrary variable $x_i$, using Lemma~\ref{lem-onecand}, we know $R$ contains exactly one candidate (say $c_i$) from the $x_i$-gadget.
We set $\pi(x_i)$ as follows:
\begin{itemize}
    \item If $c_i$ is above the $x$-axis, we set $\pi(x_i) = \mathsf{true}$.
    \item If $c_i$ is below the $x$-axis, we set $\pi(x_i) = \mathsf{false}$.
\end{itemize}
To complete the proof, we need to show that $\pi$ is a satisfying assignment of $\phi$.
It is enough to show that for each clause, at least one of its variables is set to $\mathsf{true}$.
Since the argument for positive and negative clauses is similar, we will only that for each positive clause, at least one its variables is set to $\mathsf{true}$.
We begin by proving the following important structural property of the committee $R$.

\begin{lemma} \label{lem-true}
For a piece $s$ above the $x$-axis that is not adjacent to any clause reference point, suppose $x_i$ is the associated variable of $s$.
If $R$ contains the candidate in the $s$-gadget which is $0.4$ above and $0.3$ to the left of $s^+$, then $\pi(x_i) = \mathsf{true}$.
\end{lemma}
\begin{proof}
Let $s$ be piece as described in the above lemma with the associated variable $x_i$.
We will show that if $R$ contains the candidate in the $s$-gadget which is $0.4$ above and $0.3$ to the left of $s^+$, then $R$ contains a candidate $0.2$ above $\hat{x_i}$ in the $x_i$-gadget.
The choice of candidate in $R$ from the $s$-gadget to the $x_i$-gadget percolates as follows:
\begin{itemize}
    \item When $s$ is not adjacent to a variable reference point $\hat{x_i}$: Let $v_1, v_2, v_3$, and $v_4$ be the voters placed at distances $0.4$ above and $0.3$ left of $s^+$, $0.3$ above $s^+$, at $s^+$, and $0.3$ below and $0.3$ to the right of $s^+$, respectively.
    We know that the candidate placed at $s^+$ belongs to $J$.
    From the set of $f+1$ candidates placed at the locations of $v_1, v_2$ and $v_4$, we denote a candidate by $c_1, c_2$, and $c_3$, respectively.
    Moreover, let $s'$ be the piece below (resp., to the left of) $s$ when $s$ is a vertical (resp., horizontal) piece.
    We denote the corresponding candidate in $s'$ by $c_1', c_2'$, and $c_3'$.
    Assume that $c_1 \in R$.
    We observe that $d(v_4, c_1) = \sqrt{0.85} > r$ but $d(v_4, c_1') = \sqrt{0.45} < r$. 
    Moreover, $c_1'$ is the only alive candidate from $s'$-gadget within distance $r$ from $v_4$.
    Using Lemma~2, we know $R$ only includes $c_1$ from $s$.
    Hence, to satisfy $d(v_4, R) \leq r$, $R$ must include the candidate $c_1'$.
    Observe that we can repeat the above argument for all pieces below (resp., to the left of) $s$, which implies that for all pieces $s_i$ below (resp., to the left of) $s$, $R$ includes the candidate in $s_i$-gadget placed $0.4$ above and $0.3$ to the left of $s_i^+$.
    
    \item When $s$ is adjacent to a variable reference point $\hat{x_i}$: Let $v_1, v_2, v_3$, and $v_4$ be the voters placed at distances $0.4$ above and $0.3$ left of $s^+$, $0.3$ above $s^+$, at $s^+$, and $0.3$ below and $0.3$ to the right of $s^+$, respectively.
    Moreover, let $c_1, c_2$, and $c_3$ denote an arbitrary candidate placed at the locations of voters $v_1, v_2$, and $v_4$, respectively.
    We know that for voter $v_4$, the set of candidates within distance $r$ is $\{c_2, c_3, c_4\}$ where $c_4$ is the candidate placed $0.2$ above $\hat{x}_i$.
    Using Lemma~\ref{lem-onecand}, we know $R$ only includes ${c}_1$ from the piece $\hat{s}$.
    Hence, to ensure $d(v_4, R) \leq r$, $R$ must include candidate $c_4$.
    Recall that for a variable $x_j$ for $j \in [n]$, if $R$ includes a candidate above the reference point $\hat{x}_j$, we set $x_j = \mathsf{true}$.
    Since, $c_4 \in R$ and $c_4$ lies above $\hat{x}_i$, we set $x_i = \mathsf{true}$.
    
\end{itemize}
This completes the proof of Lemma~\ref{lem-true}.
\end{proof}

Using Lemma~\ref{lem-onecand} and Lemma~\ref{lem-true} we are now ready to show that the constructed assignment $\pi$ satisfies $\phi$.
Since the argument for positive and negative clauses is similar, we will only show for the positive clauses.
Our argument is similar to the one in \cite{sonar2022multiwinner} but we include it here for completeness.

Consider a positive clause $z_i$.
We will show that at least one variable of $z_i$ is set to $\mathsf{true}$ by $\pi$.
We denote the pieces adjacent to the reference point $\hat{z_i}$ by $s_1, s_2, s_3$.
Without loss of generality, let $s_1$ be to the left of $\hat{z}_i$, $s_2$ be to the right of $\hat{z}_i$, and $s_3$ be below $\hat{z}_i$.
Notice that $\hat{z}_i = s_1^+ = s_2^+ = s_3^+$.
Recall that $\hat{z_i}$ is a connection point.
Since all connection points have even coordinates and $s_1^-,s_2^-,s_3^-$ are at a unit distance from $\hat{z_i}$, $s_1^-,s_2^-,s_3^-$ are not connection points.
Hence, let $s_4,s_5,s_6$ be the pieces such that the right endpoint of $s_4$ is $s_1^-$, the left endpoint of $s_5$ is $s_2^-$, and the top endpoint of $s_6$ is $s_3^-$.
Therefore, $s_1^- = s_4^+$, $s_2^- = s_5^+$, and $s_3^- = s_6^+$.
Let $c_4$ be a candidate in $s_4$-gadget such that $c_4$ is $0.4$ above and $0.3$ to the left of $s_4^+$.
Moreover, let the candidates $c_5$ and $c_6$ defined in the similar way such that $c_5$ belongs to the $s_5$-gadget and $c_6$ belongs to the $s_6$-gadget.
For the voter at the reference point $\hat{z_i}$, only the candidates $c_4, c_5, c_6$ are within distance $r$.
This is because all pieces except $s_1,\dots,s_6$ have distances at least $2$ from $\hat{z}_i$.
Since $\sigma(R) \leq r$, $R$ must contain at least one of these three candidates.
Therefore, using Lemma~\ref{lem-true}, we conclude that at least one of the associated variables of $s_4,s_5,s_6$ is true.
Since these are exactly the three variables in clause $z_i$; hence, $z_i$ is true under the assignment $\pi$.
This completes the ``only if'' part of our proof.

The argument above completes the proof of Lemma~\ref{lem:ORP-equivalence} and completes the argument of equivalence for our decision problem $(i)$.

\paragraph*{Argument of equivalence for the decision problem $(ii)$:}
Recall that for the FTS decision problem instance stated above, the input committee is $T = C_{cov}$ and the fault-tolerance parameter is $f = k$.
Notice that $T$ contains one candidate from each vertex and each piece gadget.
Suppose a subset $J \subset C$ with $|J| \leq f$ fails.
Consider the committee $T' = T \setminus J$.
All voters in the vertex gadgets and the piece gadgets which have a non-empty intersection with $T'$ have a committee member within distance $r$ (i.e., suppose $V' \subseteq V$ is the subset of all such voters then $\sigma(V', T') \leq r$).
For the voters in $V \setminus V'$, we build a committee $R$ using the candidates in vertex and piece gadgets with have an empty intersection with $T'$ in the same way as we constructed a replacement committee in the forward direction of proof of Lemma~\ref{lem:ORP-equivalence} (Section~\ref{subsec:ORP-equivalence}) (to replace the candidates from set $T \cap J$).
Notice that it is always possible to construct such a replacement $R$ because all candidates in $R$ are robust candidates (meaning that there are $f+1$ identical copies of each of these candidates) and the failure set $J$ has size at most $f$.

Since the total number of vertex and edge gadgets is $k$, the new committee $(T\setminus J) \cup R$ has size $k$.
Using the same argument as in Section~\ref{subsec:ORP-equivalence}, we can show that $\sigma((T\setminus J) \cup R) \leq r$.
This completes the argument in the forward direction (that is, if $\phi$ is satisfiable then $\sigma_f(T) \leq r$).

To show the reverse direction, suppose $\sigma_f(T) \leq r$.
Therefore, when failing set $J = C_{cov}$, there exists a $k$-sized replacement $R \subseteq C_{rob}$ such that $\sigma(R) \leq r$.
Hence, using Lemma~\ref{lem:ORP-equivalence}, we can conclude that $\phi$ is satisfiable.

\paragraph*{Argument of equivalence for the decision problem $(iii)$:}
Recall that for the OFTC decision problem instance stated above, the input is the committee size $k = N + n - 3m$ and the fault-tolerance parameter $f = k$.

The argument for the forward direction is trivial, as we know when $\phi$ is satisfiable, the $k$-sized committee $T = C_{cov}$ has $\sigma_f(T) \leq r$.
In the reverse direction, suppose $T \subseteq C$ is a $k$-sized committee with $\sigma_f(T) \leq r$.
Therefore, in this case, for a size $f$ failing set $J = C_{cov}$, there exists a replacement $R \subseteq C\setminus C_{cov}$ such that $|(T\setminus J) \cup R| \leq k$ and $\sigma((T\setminus J) \cup R) \leq r$. 
Since $(T \setminus J) \cup R = R$ and $R \subseteq C_{rob}$; using Lemma~\ref{lem:ORP-equivalence}, we can conclude that $\phi$ is satisfiable.

\subsubsection[bounded f proofs]{Hardness when $f$ is bounded}
\label{subsec:bounded-f}

\paragraph*{ORP and FTS}
Consider an election $E = (V,C)$, a committee $T \subseteq C$ of size $k$ and a fault-tolerance parameter $f$ which is a constant.
It is easy to see that we can solve ORP optimally in time $nm^{\mathcal{O}(f)}k$ by trying all possible replacement sets and choosing the best one.
Similarly, by trying all possible failing sets of size at most $f$ (note that there are $m^{\mathcal{O}(f)}$ such sets) and computing optimal replacement for each set, we can compute $\sigma_f(T)$ in time $nm^{O(f)}k$.

\paragraph*{OFTC}
We will now show that for any integer $f \geq 0$, OFTC is NP-hard using a simple reduction from the $k$-supplier problem \cite{nagarajan2013euclidean}:
Fix $f \geq 0$.
Let $(\mathcal{C}, F)$ along with an integer $k$ be a $k$-supplier instance where $\mathcal{C}$ is the set of customers and $F$ is a set of facilities embedded in $\mathbb{R}^2$.
In the decision version of $k$-supplier, given a real number $r$, we ask if there exists a size $k$ set $F' \subseteq F$ such that $\sigma(\mathcal{C}, F') \leq r$.

We construct an election $E = (V, C)$ in $\mathbb{R}^2$ as follows. 
We set $V = \mathcal{C}$.
Further, we construct the set of candidates $C$ by adding $f+1$ identical candidates on each point in $F$.
We set the committee size to $k$.
It is easy to see that there exists a $k$-sized committee $T \subset C$ with $\sigma_f(T) \leq r$ if and only if there exists a $k$-sized subset $F' \subseteq F$ such that $\sigma(\mathcal{C}, F') \leq r$ where $r$ is a real number.

\subsection{Optimal Replacement Problem}
\label{subsec:ORP-2d}
A simple greedy algorithm achieves a $3$-approximation for the Optimal Replacement Problem
in any fixed dimension $d$ as well as in any metric space.

\begin{lemma} \label{lem:ORP-2d-3-approx}
    We can find a $3$-approximation for ORP in time $O(k(nk+m))$.
\end{lemma}
\begin{proof}
Let $T \subseteq C$ be the given committee and let $J \subseteq T$ be the failing set. 
In order to find the replacement set $R$, we initialize $\hat{T} = T\setminus J$, and then
repeat the following two steps $|T \cap J|$ times: (1) Choose the farthest voter from $\hat{T}$, namely, choose 
$\hat{v} = \argmax_{v \in V} d(v, \hat{T})$, and
(2) Add to $\hat{T}$ the candidate $\hat{c} \notin \hat{T}$ that is closest to $\hat{v}$.
Upon termination, we clearly have  $|\hat{T}| = |T|$. 

We will now show that $R$ is a $3$-approximate replacement committee.
    Let $|T\cap J| = r$. %and let $V'$ be the set of voters with their closest candidate in $T$ is a candidate in set $T\cap J$.
    Suppose $R^* = \{c_1^*, c_2^*, \ldots, c_r^*\}$ is an optimal replacement such that $T^* = (T\setminus J)\cup R^*$ has $\sigma(T^*) = \sigma^*$.
    Let $V_1^*, V_2^*, \ldots, V_r^* \subseteq V$ be disjoint set of voters such that $c_i^*$ is the closest candidate in $T^*$ for all voters in $V_i^*$ for $i \in [r]$.
    We define $V' = \bigcup_{i=1}^r V_i^*$.
   
    Let $R = \{c_1, c_2, \ldots, c_r\}$ be the replacement set constructed by our algorithm and let $\hat{v_1}, \hat{v_2}, \ldots, \hat{v_r}$ be the voters chosen by our algorithm.
    Recall that $\hat{T} = (T\setminus J)\cup R$.
    It is easy to see that $\sigma(V\setminus V', T^*) = \sigma(V\setminus V', \hat{T}) \leq \sigma^*$ since $ (T\setminus J) \subseteq T^*$ and $(T\setminus J) \subseteq  \hat{T}$.
    Next, using a simple case analysis we will show that $\sigma(V', \hat{T}) \leq 3\sigma^*$.

    Our cases are based on the voters $\hat{v_1}, \hat{v_2}, \ldots, \hat{v_r}$ as follows:
    
    \begin{itemize}
        \item If $\hat{v_i} \in V\setminus V'$ for $i\in [r]$, then the farthest voter (aka $\hat{v_i}$) in the $i^{th}$-iteration of the algorithm satisfies $\sigma(\hat{v_i}, \hat{T}) \leq \sigma(V\setminus V', \hat{T}) \leq \sigma^*$, and hence, $\sigma(V', \hat{T}) \leq \sigma^*$.
        \item Next, suppose for some $i,j,k \in [r]$, we have $\hat{v_i}, \hat{v_j} \in V_k^*$.
        Without loss of generality, let $j > i$.
        Then, in the $j^{th}$-iteration, the farthest voter (aka $\hat{v_j}$) among all voters in $V$ has $d(\hat{v_j}, \hat{T_{j-1}}) \leq 3\sigma^*$ where $\hat{T_{j-1}}$ is a committee after $j-1$ iterations of adding replacement candidates.
        This is because $d(\hat{v_j}, \hat{T_{j-1}}) \leq d(\hat{v_j}, \hat{v_i}) + d(\hat{v_i}, \hat{c_i})$, and we know $d(\hat{v_j}, \hat{v_i}) = d(\hat{v_j}, c_i^*) +d(c_i^*, v_i) \leq 2\sigma^*$ and $d(\hat{v_i}, \hat{c_i}) \leq \sigma^*$ since $c_i$ is the closest candidate to $v_i$ which is not in the committee.
        This implies that $\sigma(\hat{T_{j-1}}) \leq 3\sigma^*$.
        Since $\hat{T_{j-1}} \subseteq \hat{T}$; therefore, $\sigma(V', \hat{T}) \leq 3\sigma^*$.
        \item Finally, we consider the case for when $i \in [r]$, we have $\hat{v_i} \in V_i^*$.
        We will now show that for $v \in V_i^*$, $\sigma(v, \hat{T}) \leq 3\sigma^*$.
        We know $d(v, \hat{c_i}) \leq d(v, \hat{v_i}) + d(\hat{v_i}, \hat{c_i})$.
        Notice that $d(v, \hat{v_i}) \leq d(v, c_i^*) + d(c_i^*, \hat{v_i}) \leq 2\sigma^*$, and $d(\hat{v_i}, \hat{c_i}) \leq \sigma^*$.
        Therefore, $d(v, \hat{c_i}) \leq 3\sigma^*$.
        This implies, $\sigma(V', R) \leq 3\sigma^*$, and hence, $\sigma(V', \hat{T}) \leq 3\sigma^*$.
    \end{itemize}

    Finally, it is easy to see that the algorithm runs for at most $k$ iterations and each iteration can be trivially implemented in time $O(nk+m)$.
    This completes the proof of Lemma~\ref{lem:ORP-2d-3-approx}.
    
\end{proof}

\subsection{Computing the Fault-Tolerance Score} \label{subsec:FTS-2d}

We can also approximate the optimal fault-tolerance score of a committee within a factor of $3$.
Specifically, if the optimal fault-tolerance score of $T$ is $\sigma_f(T) = \sigma^*$, then our algorithm 
returns a real number $\sigma'$ such that $\sigma^* \leq \sigma' \leq 3\sigma^*$.

For each voter $v$, let $d_f(v)$ be $v$'s distance to its $(f+1)^{th}$ closest candidate, and let $d' = \max_{v \in V} d_f(v)$
be the maximum of these values over all voters. The basic idea behind our approximation is simple and uses the following two facts:
(1) $\sigma^* \geq d'$, and (2) $\sigma^* \geq \sigma(T)$.
The first one holds because $d'$ is the best score possible if some voter's $f$ nearest candidates fail, and
the second one holds because a failure can only worsen the score (that is, $\sigma_f(T) \geq \sigma(T)$ for any $f > 0$).
Therefore, the distance $\sigma' = d' + 2\sigma(T)$ is clearly within a factor of $3$ of the optimal $\sigma^*$.
We claim that for any failing set $J \subseteq C$, there exists a replacement $R \subseteq C\setminus J$ of 
size at most $|T\cap J|$ such that $\sigma((T\setminus J)\cup R) \leq \sigma'$.
%Due to limited space, we omit the proof from the extended abstract; it is included in the appendix
%(see Claim~\ref{clm:replacement-3-approx}).

\begin{claim}
\label{clm:replacement-3-approx}
For a committee $T \subseteq C$ and a failing set $J \subseteq C$, there exists a replacement $R \subseteq C\setminus J$ of size at most $|T\cap J|$ such that $\sigma((T\setminus J)\cup R) \leq \sigma'$ where $\sigma' = d'+2\sigma(T)$.
\end{claim}
\begin{proof}
Let $T \cap J = \{c_1, \ldots, c_r\}$ and let $V_1, \ldots, V_r$ be disjoint sets of voters such that $c_i$ is the closest candidate in $T$ to all voters in $V_i$ for $i \in [r]$.
We define $\overline{V} = \bigcup_{i=1}^r V_i$ and $V' = V\setminus \overline{V}$.
For each voter $v \in V'$, its closest candidate in $T$ is still available; hence, $\sigma(v, T\setminus J) \leq \sigma(T)$. 
Therefore, we only need to show that $\sigma(\overline{V}, R) \leq \sigma'$.

We build a replacement $R$ as follows:
We initialize $R = \emptyset$.
Next, for $r$ iterations, let $i$ be the iteration index then,

\begin{itemize}
    \item Select an arbitrary voter $v_i \in V_i$.
    \item Let $\hat{c_i} \in C\setminus J$ be the closest available candidate to $v_i$. 
    Add $\hat{c_i}$ to $R$.
\end{itemize}

To show $\sigma(\overline{V}, R) \leq \sigma'$, we will now show that for all $i\in [r]$, $\sigma(V_i, \hat{c_i}) \leq \sigma'$.
For $v \in V_i$, we know $d(v, \hat{c_i}) \leq d(v, v_i) + d(v_i, \hat{c_i})$ where $v_i$ is the voter selected in the $i^{th}$-iteration. 
First, observe that $d(v_i, \hat{c_i}) \leq d'$ as at most $f$ closest candidates to $v_i$ belong to $J$.
Second, $d(v, v_i) \leq d(v, c_i) + d(c_i, v_i) \leq 2\sigma(T)$ (recall that $c_i \in T \cap J$).
Hence, $d(v, \hat{c_i}) \leq 2\sigma(T) + d' = \sigma'$.
Therefore, for the constructed replacement $R$, $\sigma(\overline{V}, R) \leq \sigma'$.
\end{proof}

We have the following result.

\begin{lemma} \label{lem:FTS-2d-3-approx}
    The fault-tolerance score of a committee can be approximated within a factor of $3$ in time $\mathcal{O}(nm \log(f))$.
\end{lemma}
\subsection{Optimal Fault-Tolerant Committee}
\label{subsec:OFTC-2d}
We now discuss how to design approximately optimal fault-tolerant committees in multiwinner elections.
Specifically, given a set of voters $V$ and a set of candidates $C$ in $d$-space, along with parameters $k$ (committee size) and 
$f$ (number of faults), we want to compute a size $k$ committee $T \subseteq C$ with the minimum fault-tolerance score 
$\sigma_f(T)$. We prove two approximation results for this problem:
(1) We can solve this problem in polynomial time within an approximation factor of $3$ in polynomial time if the parameter $f$ is treated as a constant 
	(while $k$ remains possibly unbounded). If $f$ is not assumed to be a constant, we can solve the
	problem within an approximation factor of $5$.
(2) We give an EPTAS with running time $(1/\varepsilon)^{O(1/{\varepsilon^{2d}})}(m+n)^{O(1)}$ which is a
	\emph{bicriterion} approximation, where the output committee $T$ is fault-tolerant for at least 
	$(1-\varepsilon)n$ voters with $\sigma_f(T) \leq (1+\varepsilon)\sigma^*$.
The next two subsections discuss these results.

\subsubsection[3-approx-bounded-f]{3-Approximation for Bounded $f$}
\label{sec:OFTC-bounded-f}
Let $\sigma^*$ be the optimal $f$-tolerant score of a committee of size $k$.  
We compute the approximation solution via an approximate decision algorithm, which takes as input a number $\sigma \geq \sigma_f(C)$ and returns a committee $T \subseteq C$ of size at most $k$ with $\sigma_f(T) \leq 3 \sigma$ if $\sigma \geq \sigma^*$. 
(We slightly abuse notation to introduce a convenient quantity $\sigma_f(C)$, which is the $f$-fault-tolerance score of a committee with all the input candidates. This is clearly a lower bound on any size $k$ committee's score.)

For a committee $T \subseteq C$ and a failing set $J \subseteq C$, let $\delta(T, J)$ denote the score obtained 
after finding an optimal replacement $K$. That is,
\begin{equation*}
    \delta(T,J) = \min_{K \in C \backslash J, |K| = |T \cap J|} \sigma_0((T \backslash J) \cup K).
\end{equation*}
Thus, $\sigma_f(T) = \max_{J \subseteq C, |J| \leq f} \delta(T,J)$.
Our approximation algorithm is shown in Algorithm~1. It begins with an empty committee $T$ (line~1),
and as long as there exists a failing set $J$ of size at most $f$ for which $\delta(T,J) > 3 \sigma$, %(note that, we can check this condition by iterating over all $f$-sized failing sets and computing an optimal replacement set in each case)
\footnote{We can check this condition by iterating over all failing sets of size $f$ and computing an optimal replacement set in each case.} 
we do the following.

First, we remove all candidates in $J$ from $T$ (line~3).  
Then, whenever there exists a voter $v \in V$ with $d(v,T) > 3 \sigma$, 
we add to $T$ a candidate $c \in C \backslash J$ whose distance to $v$ is at most $\sigma$ (lines~5-6).
Such a $c$ always exists because $\sigma$ is at least the distance to the $(f+1)^{th}$ closest neighbor to $v$. 
%\begin{equation*}
    % $\sigma \geq \sigma_f(C) \geq \delta(C,J) = \sigma_0(C \backslash J)$.
%\end{equation*}

We call this voter $v$ the \textit{witness} of $c$, denoted by $\mathsf{wit}[c]$~(line~7).
Adding $c$ to $T$ guarantees that $d(v,T) \leq \sigma$.
We repeat this procedure (the inner while loop) until $d(v,T) \leq 3 \sigma$ for all $v \in V$.
Finally, the outer while loop terminates when $\delta(T,J) \leq 3 \sigma$ for all $J \subseteq C$ of size at 
most $f$, i.e., $\sigma_f(T) \leq 3 \sigma$.
At this point, we return the committee $T$.

\begin{algorithm}
\caption{Approximate decision algorithm}
\label{alg:testing}
\begin{algorithmic}[1]
\Statex \textbf{Input:} a set $V$ of voters, a set $C$ of candidates, the committee size $k$, the fault-tolerance parameter $f$, and a number $\sigma \geq \sigma_f(C)$
\State $T \gets \emptyset$
\While{$\exists$ $J \subseteq C$ such that $|J| \leq f$ and $\delta(T,J) > 3 \sigma$}
    \State $T \leftarrow T \backslash J$
    \While{$\exists$ $v \in V$ such that $d(v,T) > 3 \sigma$}
        \State $c \gets$ a candidate in $C \backslash J$ satisfying $d(v,c) \leq \sigma$
        \State $T \gets T \cup \{c\}$
        \State $\mathsf{wit}[c] \gets v$
    \EndWhile
\EndWhile
\State \textbf{return} $T$
\end{algorithmic}
\end{algorithm}

% l

\begin{restatable}[]{lemma}{faraway}
\label{lem-faraway}
Let $T$ be the committee computed by Algorithm~\ref{alg:testing}.
Then $d(\mathsf{wit}[c],\mathsf{wit}[c']) > 2 \sigma$ for any two distinct $c,c' \in T$. 
%(See Figure~\ref{fig:3-approx-const-f} for intuition.)
\end{restatable}
\begin{proof}
Let $c,c' \in T$ such that $c \neq c'$.
When the committee $T$ is constructed in Algorithm~\ref{alg:testing}, the candidates are added to $T$ one by one (line~6).
Therefore, without loss of generality, we can assume that $c'$ is added to $T$ after $c$.
Consider the iteration of the inner while-loop (line~4-7) of Algorithm~\ref{alg:testing} in which we add $c'$ to $T$.
At the beginning of this iteration, we have $d(v,T) > 3 \sigma$ where $v = \mathsf{wit}[c']$.
Note that $c \in T$ at this time, and thus $d(\mathsf{wit}[c'],c) > 3 \sigma$.
Furthermore, we have $d(\mathsf{wit}[c],c) \leq \sigma$ by construction.
Therefore,
\begin{equation*}
    d(\mathsf{wit}[c],\mathsf{wit}[c']) \geq d(\mathsf{wit}[c'],c) - d(\mathsf{wit}[c],c) > 2 \sigma,
\end{equation*}
by the triangle inequality.
\end{proof}

\begin{lemma} \label{lem-3apprx}
If $\sigma \geq \sigma^*$, then Algorithm~\ref{alg:testing} outputs a size $k$ committee $T$ with $\sigma_f(T) \leq 3 \sigma$.
\end{lemma}
\begin{proof}
The condition of the outer while loop of Algorithm~\ref{alg:testing} guarantees that $\delta(T,J) \leq 3 \sigma$ for all $J \subseteq C$ of size at most $f$, which implies $\sigma_f(T) \leq 3 \sigma$.
To prove $|T| \leq k$, suppose $T = \{c_1,\dots,c_r\}$.
By Lemma~\ref{lem-faraway}, the pairwise distances between the voters $\mathsf{wit}[c_1],\dots,\mathsf{wit}[c_r]$ are all larger than $2 \sigma$ and thus larger than $2 \sigma^*$ (as $\sigma \geq \sigma^*$ by our assumption).
%Since $\sigma \geq \sigma^*$ by assumption, for any candidate $c \in C$, there exists at most one index $i \in [r]$ such that $d(\mathsf{wit}[c_i],c) \leq \sigma^*$.
%Indeed, if there exist two indices $i,j \in [r]$ such that $d(\mathsf{wit}[c_i],c) \leq \sigma^*$ and $d(\mathsf{wit}[c_j],c) \leq \sigma^*$, then $d(\mathsf{wit}[c_i],\mathsf{wit}[c_j]) \leq 2 \sigma^* \leq 2 \sigma$, contradicting the fact that $d(\mathsf{wit}[c_i],\mathsf{wit}[c_j]) > 2 \sigma$.
Now consider a committee $T^* \subseteq C$ of size $k$ satisfying $\sigma_f(T^*) = \sigma^*$.
For each $\mathsf{wit}[c_i]$, there exists $c_i^* \in T^*$ such that $d(\mathsf{wit}[c_i],c_i^*) \leq \sigma^*$.
Observe that $c_1^*,\dots,c_r^*$ are all distinct.
Indeed, if $c_i^* = c_j^*$ and $i \neq j$, then by the triangle inequality,
\begin{equation*}
    d(\mathsf{wit}[c_i],\mathsf{wit}[c_j]) \leq d(\mathsf{wit}[c_i],c_i^*) + d(\mathsf{wit}[c_j],c_j^*) \leq 2 \sigma^*,
\end{equation*}
contradicting the fact that $d(\mathsf{wit}[c_i],\mathsf{wit}[c_j]) > 2 \sigma^*$.
Since $|T^*| = k$ and $c_1^*,\dots,c_r^* \in T^*$, we have $r \leq k$, which completes the proof.
\end{proof}

Using these two lemmas, we can compute a $3$-approximate solution using Algorithm~\ref{alg:testing} as follows.
First, we compute $\sigma_f(C)$ in $O(nm^{f+1})$ time by enumerating all failing sets $J \subseteq C$ of size at most $f$.
For every voter $v \in V$ and every candidate $c \in C$ such that $d(v,c) \geq \sigma_f(C)$, we run Algorithm~\ref{alg:testing} with $\sigma = d(v,c)$.
Among all the committees returned of size at most $k$, we pick the one, say $T^*$, that minimizes $\sigma_f(T^*)$.
To see that $\sigma_f(T^*) \leq 3 \sigma^*$, note that $\sigma^*$ must be the distance between a voter and a candidate.
Thus, there is one call of Algorithm~\ref{alg:testing} with $\sigma = \sigma^*$, which returns a committee $T \subseteq C$ of size at most $k$ such that $\sigma_f(T) \leq 3 \sigma = 3 \sigma^*$, by Lemma~\ref{lem-3apprx}.
We have $\sigma_f(T^*) \leq \sigma_f(T)$ by construction, which implies $\sigma_f(T^*) \leq 3 \sigma^*$.
% In Appendix~\ref{sec:runtime-alg1}, we show that each call of Algorithm~\ref{alg:testing} takes $O(nm^{2f+1})$ time.
% We need to call the algorithm $O(nm)$ times.

\emph{Overall running time.}
We will show that each run of Algorithm~\ref{alg:testing} takes $O(nm^{2f+1})$ time.
We can check the condition of while loop in Step 2 in time $O(m^{2f})$.
This is because there are at most $O(m^f)$ failing sets of size at most $f$ (the precise upper bound is $2m^f$), and for each failing set, we can find an optimal replacement in time $O(m^f)$ by bruteforce.
Next, each iteration of the while loop takes $O(nm)$ time, that is, the time required to compute all voter-candidate pairwise distances.
Therefore, the overall running time of the algorithm is $O(n^2m^{2f+2})$.

Thus, we have the following result.

\begin{theorem}
\label{thm:OFTC-2d-3-approx}
We can find a $3$-approximation for Optimal Fault-tolerant Committee in time $O(n^2 m^{2f+2})$, assuming the 
fault-tolerance parameter $f$ is a constant.
\end{theorem}

Next, for a non-constant $f$, we give a $5$ approximation using a greedy rule.

\begin{restatable}[]{lemma}{OFTCapprox}
\label{lem:OFTC-2d-5-approx}
We can find a $5$-approximation for Optimal Fault-Tolerant Committee in time $\mathcal{O}(mnk)$.
\end{restatable}
%We note that our algorithm for OFTC holds in any metric space.
Our algorithm is quite simple and uses the classical ``farthest next'' greedy rule \cite{gonzalez1985clustering}.
Specifically, let $C$ and $V$ be the set of candidates and voters, respectively.
We begin with an empty committee $T = \emptyset$ and an empty set $\hat{V}$ of picked voters.
Then we repeat the following step: pick the voter $\hat{v} \in V \backslash \hat{V}$ farthest to the current committee $T$, add $\hat{v}$ to $\hat{V}$, and add to $T$ the candidate $\hat{c} \in C$ closest to $\hat{v}$.
The procedure terminates when $|T| = k$ or the candidate $\hat{c}$ computed is already in $T$.
Formally, our algorithm is shown in Algorithm~\ref{alg:5-approx}.

\begin{algorithm}
\caption{$~~~5$-approximation algorithm for OFTC}
\label{alg:5-approx}
\begin{algorithmic}[1]
\Statex \textbf{Input:} a set $V$ of voters, a set $C$ of candidates, the committee size $k$, and the fault-tolerance parameter $f$
%\State \textbf{Output:} An $f$-fault-tolerant committee $T$ with $|T| \leq k$.
\State $i \gets 0$ and $T \gets \emptyset$
% \NoNumber \State Suppose $J \subset C$ s.t. $|J| \leq f$ is a failing set and let $R$ is the smallest replacement for $J$ s.t. $\sigma_0((T\setminus J) \cup R) \leq 3\sigma^*$
\While{$|T| \leq k$}
    \State $i \gets i+1$
    \State $v_i \gets \arg \max_{v \in V \backslash \hat{V}} d(v, T)$
    %\State $\hat{V} \gets \hat{V} \cup \{\hat{v}\}$
    \State $c_i \gets \arg \min_{c \in C} d(v_i,c)$
    \If{$c_i \in T$}
        \State \textbf{break}
    \EndIf
    \State $T \gets T \cup  \{c_i\}$
\EndWhile
\State \textbf{return} $T$
\end{algorithmic}
\end{algorithm}

We now move on to the proof of correctness.
Denote by $\sigma^*$ the optimal $f$-tolerant score of a size-$k$ committee.
First, using the same analysis as the one for the $k$-center problem \cite{hochbaum1986unified}, we can show that $\sigma(T) \leq 3 \sigma^*$.

\begin{lemma}\label{lem:k-supplier}
Let $T$ be the committee computed by Algorithm~\ref{alg:5-approx}.
Then $\sigma(T) \leq 3 \sigma^*$.    
\end{lemma}
The proof Lemma~\ref{lem:k-supplier} is easy and we refer the reader to \cite{hochbaum1986unified} for details.
We now show that $T$ is a $5$-approximate solution to OFTC.

\begin{lemma}
Let $T$ be the committee computed by Algorithm~\ref{alg:5-approx}.
Then $\sigma_f(T) \leq 5 \sigma^*$.
\end{lemma}
\begin{proof}
It suffices to show that for any failing set $J \subseteq C$ of size at most $f$, there exists a replacement set $K \subseteq C \backslash J$ such that $|K| = T \cap J$ and $\sigma_0((T \backslash J) \cup K) \leq 5 \sigma^*$.
Suppose $T = \{c_1,\dots,c_r\}$, where $c_i$ is the candidate selected in the $i$-th iteration of Algorithm~\ref{alg:5-approx}.
Let $v_1,\dots,v_r$ be the voters computed in line~4 of Algorithm~\ref{alg:5-approx}.
For a failing set $J \subseteq C$, we construct the replacement set $K$ as follows: For each index $i \in [r]$ such that $c_i \in J$, we include in $K$ the candidate $c_i' \in C \backslash J$ closest to $v_i$. Clearly, $|K| = |T \cap J|$.

Now we show that $\sigma_0((T \backslash J) \cup K) \leq 5 \sigma^*$.
Using Lemma~\ref{lem:k-supplier}, we know that $d(v, T) \leq 3\sigma^*$ for any voter $v \in V$.
Based on this, we bound $\sigma_0((T \backslash J) \cup K)$ as follows.
Observe that $\sigma_0((T \backslash J) \cup K) = \max_{v \in V} d(v,(T \backslash J) \cup K)$.
So it suffices to show that $d(v,(T \backslash J) \cup K) \leq 5 \sigma^*$ for all $v \in V$.
Let $c_i \in T$ be the candidate closest to $v$; thus, $d(v,c_i) = d(v,T) \leq 3 \sigma^*$.
If $c_i \notin J$, we are done.
Otherwise, $c_i' \in K$ and hence $d(v,(T \backslash J) \cup K) \leq d(v,c_i')$.
By the triangle inequality, we have
\begin{equation*}
    d(v,c_i') \leq d(v,c_i)+d(c_i,v_i)+d(v_i,c_i').
\end{equation*}
As argued before, $d(v,c_i) \leq 3 \sigma^*$.
Furthermore, $d(c_i,v_i) \leq d(v_i,c_i') \leq \sigma^*$, because $c_i'$ is the candidate in $C \backslash J$ closest to $v_i$.
Therefore, the above inequality implies $d(v,c_i') \leq 5 \sigma^*$.
\end{proof}

By the above lemma, we know that Algorithm~\ref{alg:5-approx} achieves an approximation ratio of 5.
Its running time is clearly $O(mnk)$.
This completes the proof of Lemma~\ref{lem:OFTC-2d-5-approx}.

All of these approximations hold not just for $d$-dimensional Euclidean space, for any fixed $d$, but also for any metric space.

\subsubsection{A bicriterion EPTAS}
\label{sec:OFTC-EPTAS}

Finally, we design a bicriterion FPT approximation scheme with running time $f(\varepsilon) \cdot n^{\mathcal{O}(1)}$, which finds a size-$k$ committee whose fault-tolerance score for at least a $(1 - \varepsilon)$ fraction of the voters is within a factor of $(1 + \varepsilon)$ of the optimum.
Formally, we say a committee $T$ is \textit{$(r,\rho)$-good} if there exists a subset $V' \subseteq V$ of size at least $\rho n$ such that the $f$-tolerant score of $T$ with respect to only the voters in $V'$ is at most $r$.
Then our approximation scheme can output a size-$k$ committee which is $((1+\varepsilon) \sigma^*, 1-\varepsilon)$-good.
The core of our approximation scheme is the following (approximation) decision algorithm.
The decision algorithm takes the problem instance and an additional number $r > 0$ as input.
The output of the algorithm has two possibilities: it either \textbf{(i)} returns YES and gives a size-$k$ committee that is $((1+\varepsilon)r,1-\varepsilon)$-good or \textbf{(ii)} simply returns NO.
Importantly, the algorithm is guaranteed to give output \textbf{(i)} as long as $r \geq \sigma^*$.
Note that this decision algorithm directly gives us the desired approximation scheme.
Indeed, we can apply it with $r = d(v,c)$ for all $v \in V$ and $c \in C$.
Let $r^*$ be the smallest $r$ that makes the algorithm give output \textbf{(i)}.
The size-$k$ committee $T^*$ obtained when applying the algorithm with $r^*$ is $((1+\varepsilon)r^*,1-\varepsilon)$-good.
We have $r^* \leq \sigma^*$ because the algorithm must be applied with $r = \sigma^*$ at some point and it is guaranteed to give output \textbf{(i)} at that time.
Thus, $T^*$ is $((1+\varepsilon) \sigma^*, 1-\varepsilon)$-good, as desired.
%Thus, the smallest $r$ that makes the algorithm give output \textbf{(i)} is smaller than or equal to $\sigma^*$.

%for a score $r>0$, if there exists a size-$k$ fault-tolerant committee $T$ with $\sigma_f(T) \leq r$, then the algorithm returns YES, along with the witness committee $T$.
%On the other hand, if there is no size-$k$ committee with $f$-tolerant score at most $(1+\varepsilon)r$ for at least $(1-\varepsilon)n$ voters, then the algorithm returns NO.
%Such a decision algorithm suffices to solve our problem because for any $r \geq \sigma^*$ our algorithm will return YES.

For simplicity of exposition, we describe our decision algorithm in two dimensions.
%; the extension to higher dimensions is straightforward.
By scaling, we may assume that the given number is $r=1$.
To solve the decision problem, our algorithm uses the shifting technique \cite{shifting}.
Let $h$ be an integer parameter to be determined later.
For a pair of integers $i,j \in \mathbb{Z}$, let $\Box_{i,j}$ denote the $h \times h$ square $[i, i+h] \times [j, j+h]$.
A square $\Box_{i,j}$ is \textit{nonempty} if it contains at least one voter or candidate.
We first compute the index set $\widetilde{I} = \{(i,j): \Box_{i,j} \text{ is nonempty}\}$.
This can be easily done in time $\mathcal{O}((n+m)h^2)$.

Consider a pair $(x,y) \in \{0,\dots,h-1\}^2$.
Let $L_{x,y}$ be the set of all integer pairs $(i,j)$ such that $i \pmod h \equiv x$ and $j \pmod h \equiv y$.
We write $\widetilde{I}_{x,y} = \widetilde{I} \cap L_{x,y}$.
For a voter $v \in V$ and a square $\Box_{i,j}$, we say $v$ is a \textit{boundary voter} for $\Box_{i,j}$ if $v \notin [i+2, i+h-2] \times [j+2, j+h-2]$.
Furthermore, we say $v$ \textit{conflicts} with $(x,y)$ if $v$ is a boundary voter in $\Box_{i,j}$ for some $(i,j) \in \widetilde{I}_{x,y}$.
%Then the following lemma holds.

\begin{restatable}[]{lemma}{shifting}
\label{lem:shifting}
There exists a pair $(x,y) \in \{0,\ldots, h-1\}^2$ such that at most $\frac{4h-4}{h^2}\cdot |V|$ voters conflict with $(x,y)$.
\end{restatable}
\begin{proof}
A voter $v \in V$ may conflict with $(x,y)$ only if for some $(i,j) \in \widetilde{I}_{x,y}$, we have $v \in \Box_{i,j}$ but $v \notin [i+2,i+h-2] \times [j+2,j+h-2]$.
Therefore, out of the total of $h^2$ pairs $(x,y)$, $v$ can only conflict with at most $h^2 - (h-2)^2$ pairs.
% Therefore, $v$ can conflict with at most $h^2 - (h-2)^2$ pairs of $(x,y)'s$ out of the total of $h^2$ pairs.
Hence, using an averaging argument, there exists a pair $(x,y)$ with at most $\frac{h^2 - (h-2)^2}{h^2}\cdot |V|$ conflicting voters from $V$.
\end{proof}

We fix a pair $(x, y) \in \{0, \ldots, h-1\}^2$ that conflicts with the minimum number of voters.
For $(i,j) \in \widetilde{I}_{x,y}$, we define the set of (non-boundary) voters $V_{i,j} = \{v \in \Box_{i,j} : v \in [i+2, i+h-2] \times [j+2, j+h-2]\}$, and the set of candidates $C_{i,j} = \{c \in C : c \in \Box_{i,j}\}$.
%Note that for $(i,j) \in \widetilde{I}_{x,y}$, the $V_{i,j}$'s (and $C_{i,j}$'s) are disjoint, moreover, the $C_{i,j}$'s form a partition of $C$.
Note that for $(i,j) \in \widetilde{I}_{x,y}$, the $C_{i,j}$'s are disjoint and form a partition of $C$.
Next, we show an important lemma which allows our algorithm to divide our problem into smaller subproblems, solve them individually, and combine the solutions to solve the overall problem. 

\begin{restatable}[]{lemma}{localglobal}
\label{lem:local-global-FT}
Let $V_1, V_2, \ldots, V_s$ be subsets of $V$ and let $T_{1}, T_{2}, \ldots, T_{s}$ be pairwise disjoint subsets of $C$ such that $T_i$ is a fault-tolerant committee for $V_i$ with $\sigma_f(T_i) = \sigma$.
Then, $T = \bigcup_{i=1}^s T_i$ is a fault-tolerant committee of $\bigcup_{i=1}^s V_i$ with $\sigma_f(T) = \sigma$.
\end{restatable}
\begin{proof}
We will show that for any failing set $J \subseteq C$, there exists a replacement set $R$ with $|R| \leq |J \cap T|$ such that $\sigma_0((T\setminus J) \cup R) \leq \sigma$.
For $i\in [s]$, let $J_i \subseteq J$ be the restriction of $J$ to $T_i$, i.e., $J_i = J \cap T_i$.
%Let $J_1, J_2, \ldots, J_s \subseteq J$ be the projection of $J$ on $T_i$ such that $J_i = J \cap T_i$.
We know $|J| \leq f$.
Since $T_i$ is a fault-tolerant committee for $V_i$; hence, there exists a valid replacement $R_i \subseteq C \backslash J$ such that $|R_i| \leq |J_i|$ and $\sigma_0((T_i \backslash J_i) \cup R_i) \leq \sigma$.
Let $R = \bigcup_{i=1}^s R_i$ (note that the $R_i$'s need not be disjoint).
Then we have $|R| \leq \sum_{i=1}^s |J_i| \leq |J \cap T_i|$ which implies $|R| \leq |J \cap T|$, and we have $\sigma_0((T \backslash J) \cup R) \leq \sigma$.
This completes the proof of Lemma~\ref{lem:local-global-FT}.
\end{proof}

Consider a pair $(i,j) \in \widetilde{I}_{x,y}$.
Let $\overline{T}_{i,j}$ be a smallest fault-tolerant committee for $V_{i,j}$ with $\sigma_f(\overline{T}_{i,j}) \leq 1$.
We observe that any inclusion-minimal fault-tolerant committee $T_{i,j}$ for $V_{i,j}$ satisfies $T_{i,j} \subseteq C_{i,j}$.
This is because any candidate outside $C_{i,j}$ has distance more than $1+6/h$ to any voter in $V_{i,j}$ (for a large enough value of $h$).
In the next section we will show how to compute a fault-tolerant committee $T_{i,j} \subseteq C_{i,j}$ for $V_{i,j}$ such that $|T_{i,j}| \leq |\overline{T}_{i,j}|$ and $\sigma_f(T_{i,j}) \leq 1+6/h$ in $h^{O(h^4)} n^{O(1)}$ time.
Assuming we can compute the above-mentioned committee $T_{i,j}$, our overall algorithm is as follows:

\begin{enumerate}
    \item Fix a pair $(x, y) \in \{0, \ldots, h-1\}^2$ conflicting with the minimum number of voters, and set $h$ to be the smallest integer such that $(4h-4)/h^2 \leq \varepsilon$ and $6/h \leq \varepsilon$.
    \item For each pair $(i,j) \in \widetilde{I}_{x,y}$, compute $T_{i,j} \subseteq C_{i,j}$.
    \item Let $T = \bigcup_{(i,j) \in \widetilde{I}_{x,y}} T_{i,j}$.
If $|T| \leq k$, return YES (along with $T$); otherwise, return NO.
\end{enumerate}

Let $V' = \bigcup_{(i,j) \in \widetilde{I}_{x,y}} V_{i,j}$.
Since the $C_{i,j}$'s are disjoint, using Lemma~\ref{lem:local-global-FT}, we conclude that $T$ is a fault-tolerant committee for $V'$.
Furthermore, from our choice of $(x,y)$, we have $|V'| \geq (1-\varepsilon) n$.
It is easy to show that the $f$-tolerant score of $T$ with respect to the voters in $V'$ is at most $1+\varepsilon$, and in addition, if $\sigma^* \geq 1$, we have $|T| \leq k$; we give a formal argument below.
This proves correctness of our decision algorithm.
% Note that even if we remove from $T_{i,j}$ all points that are not in $C_{i,j}$, the remaining points still form a fault-tolerant committee for $V_{i,j}$ with score at most $1+6/h$.
% This is because all the points in $C \backslash C_{i,j}$ are more than $1+6/h$ distance away from any point in $V_{i,j}$ (assuming $h$ is sufficiently large).
% Thus, we may assume $T_{i,j} \subseteq C_{i,j}$.
% Since the $C_{i,j}$'s are disjoint, we can conclude that the $T_{i,j}$'s are also disjoint.
% So by Lemma~\ref{lem:local-global-FT}, $T = \bigcup_{(i,j) \in \widetilde{I}_{x,y}} T_{i,j}$ is a fault-tolerant committee for $V' = \bigcup_{(i,j) \in \widetilde{I}_{x,y}} V_{i,j}$ with $\sigma_f(T) \leq (1+6/h)$.
%
% We set $h$ to be the smallest integer such that $(4h-4)/h^2 \leq \varepsilon$ and $6/h \leq \varepsilon$.
% Using Lemma~\ref{lem:shifting} and our choice of $(x,y)$, we have $|V'| \geq (1-\varepsilon) n$.
% If the computed committee $T$ has size at most $k$, our algorithm returns YES; otherwise, it returns NO.
% To see the correctness our algorithm, recall that $\sigma^*$ is an optimum score of a fault-tolerant committee for $V$.
% If $r \geq \sigma^*$, then there exists a size-$k$ fault-tolerant committee with score $r$ for $V$, and hence for $V'$ (as $V' \subseteq V$).
% In this case, we have $|T| \leq k$ and our algorithm returns YES.
% On the other hand, if there is no size-$k$ committee whose $f$-tolerant score is at most $(1+\varepsilon)r$ for at least $(1-\varepsilon)n$ voters, we must have $|T| > k$ and thus our algorithm returns NO.
The overall algorithm takes $(1/\varepsilon)^{O(1/\varepsilon^4)}(m+n)^{O(1)}$ time.
We note that the algorithm can be directly generalized to the $d$-dimensional case with running time $(1/\varepsilon)^{O(1/{\varepsilon^{2d}})}(m+n)^{O(1)}$.
Therefore, we have the following result.
\begin{theorem}
\label{thm:OFTC-2d-EPTAS}
Given a $d$-dimensional Fault-Tolerant Committee Selection instance, we can compute a size-$k$ committee $T$ such that the $f$-tolerant score of $T$ with respect to at least $(1-\varepsilon)n$ voters is at most $(1+\varepsilon)\sigma^*$, where $\sigma^*$ is the optimal $f$-tolerant score of a size-$k$ committee (with respect to the entire set $V$). This algorithm runs in time $(1/\varepsilon)^{O(1/{\varepsilon^{2d}})}(m+n)^{O(1)}$.
\end{theorem}

\paragraph{Correctness for the Decision Algorithm.}
Recall that, in our decision algorithm, we set $h$ to be the smallest integer such that $(4h-4)/h^2 \leq \varepsilon$ and $6/h \leq \varepsilon$.
Moreover, using Lemma~\ref{lem:shifting} and our choice of $(x,y)$, we have $|V'| \geq (1-\varepsilon) n$.
If the computed committee $T$ has size at most $k$, our algorithm returns YES; otherwise, it returns NO.
To see the correctness our algorithm, recall that $\sigma^*$ is an optimum score of a fault-tolerant committee for $V$.
If $r \geq \sigma^*$, then there exists a size-$k$ fault-tolerant committee with score $r$ for $V$, and hence for $V'$ (as $V' \subseteq V$).

Recall that for $(i,j) \in \tilde{I}_{x,y}$, the computed committee $T_{i,j}$ has $|T_{i,j}| \leq |\overline{T}_{i,j}|$ where $\overline{T}_{i,j}$ is the smallest committee for $V_{i,j}$ with $\sigma_{f}(\overline{T}_{i,j}) \leq 1$.
Therefore, when $r \geq \sigma^*$, for $T = \bigcup_{(i,j)\in \tilde{I}_{x,y}} T_{i,j}$, we have $|T| \leq k$ and our algorithm returns YES.
On the other hand, if there is no size-$k$ committee whose $f$-tolerant score is at most $(1+\varepsilon)r$ for at least $(1-\varepsilon)n$ voters, we must have $|T| > k$ and thus our algorithm returns NO.
This completes the argument for the proof of correctness.

\paragraph[solving-in-a-box]{Algorithm to Compute $T_{i,j}$}
We now present the most challenging piece of our algorithm: the computation of the $T_{i,j}$'s.
Consider a box $\Box_{i,j}$.
Suppose there exists a fault-tolerant committee $T \subseteq C$ for $V_{i,j}$ with $\sigma_f(T) \leq 1$.
Our task is to compute a fault-tolerant committee $T_{i,j} \subseteq C$ for $V_{i,j}$ such that $|T_{i,j}| \leq |T|$ and $\sigma_f(T_{i,j}) \leq 1+6/h$.
% The main challenge is to efficiently find the potential committees $T_{i,j}$ and to compute their near optimal $f$-tolerance scores.

We divide $\Box_{i,j}$ into $h^4$ smaller cells each with size $\frac{1}{h} \times \frac{1}{h}$, and we denote the set of these cells by $L = \{l_1, \ldots, l_{h^4}\}$. (See Figure~\ref{fig:eptas}.)
Our algorithm is based on two key observations:\\
$(i)$ A committee with a candidate in every nonempty cell has $f$-tolerant score within a difference of at most $2/h$ from the optimum score.
Since the number of cells is $h^4$, this implies that the size of a smallest approximately optimal committee is bounded by $h^4$ (formally shown in Lemma~\ref{lem-TTstar}).\\
%Given $h^4$ bound on the committee size, the next challenge is to compute their (near) optimal $f$-tolerant score efficiently.
$(ii)$ All candidates in a cell can be treated as identical, causing only a loss of $2/h$ in the score.
This implies that for any $T_{i,j}$, to approximately compute the $f$-tolerant score of $T_{i,j}$, we \emph{only} need to consider the failing sets where either all or none of the candidates in a cell fail. 
Note that the number of such failing sets is at most $2^{O(h^4)}$ (formally shown in Lemma~\ref{lem-tilde}).\\

Using these two observations, at a high level, our algorithm goes through all committees of size at most $h^4$ (there are $h^{\mathcal{O}(h^4)}$ of these as we can assume that each cell has at most $h^4$ candidates), approximately computes the $f$-tolerant score of each of these committees in time $2^{O(h^4)}$, and returns the smallest one with the desired score.

%Next, we formalize the two above observations in Lemma~\ref{lem-T*T} and Lemma~\ref{lem-tilde}, respectively.

\begin{figure}
    \centering
    \includegraphics[scale=1]{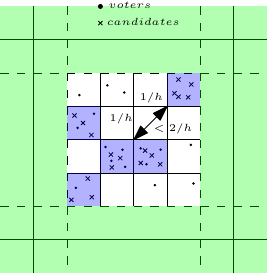}
    \caption{The figure shows a cell in the shifted grid. The solid lines around the sides are the grid lines (and the region inside them is a cell). The shaded (green) region is the boundary region. Inside the boundary region, we divide the cell into $1/h \times 1/h$ smaller cells. The distance between any two points in a smaller cell is $< 2/h$. All candidates in smaller cells are identical (i.e., candidates in blue regions). In this example, since only five cells are nonempty, we have at most $2^5$ distinct failing sets.}
    \label{fig:eptas}
\end{figure}

\begin{restatable} []{lemma}{TTstar}
\label{lem-TTstar}
Let $T, T^* \subseteq C$ be fault-tolerant committees for $V_{i,j}$.
If $|T^* \cap l_a| = 1$ for all $a \in [h^4]$ such that $C \cap l_a \neq \emptyset$, then $\sigma_f(T^*) - \sigma_f(T) \leq 2/h$.
\end{restatable}
\begin{proof}
Consider a failing set $J \subseteq C$.
Since $T$ is a fault-tolerant committee, there exists a valid replacement set $R$ such that $|(T\backslash J) \cup R| \leq |T|$ and $\sigma_0((T\backslash J) \cup R) \leq \sigma_f(T)$.
Let $L' = \{l_a \in L : l_a \cap ((T\backslash J) \cup R) \neq \emptyset\}$.
Moreover, let $J^* = J \cap T^*$.
Then we will show that there exists a replacement set $R^*$ for $J^*$ such that $|R^*| \leq |J^*|$ and $\sigma_0((T^* \setminus J^*) \cup R^*) - \sigma_0((T\setminus J) \cup R) \leq 2/h$.

We construct the set $R^*$ as follows:
Consider a cell $l_a \in L'$.
Since $l_a \cap ((T\setminus J) \cup R) \neq \emptyset$, $l_a$ is nonempty.
This implies $T^* \cap l_a \neq \emptyset$ from the way we construct $T^*$.
Let $c_a$ be the only candidate in $T^* \cap l_a$.
If $c_a \in J$, then we replace $c_a$ with an arbitrary $c_a' \in l_a\setminus J$ (i.e., we add $c_a'$ to $R^*$).
We know that such a candidate $c_a'$ exists because $l_a \cap ((T\backslash J) \cup R) \neq \emptyset$.
Since we only add a candidate to $R^*$ from $l_a$ such that $J^* \cap l_a \neq \emptyset$, we have $|R^*| \leq |J^*|$.

We will now show that $\sigma_0((T^* \setminus J^*) \cup R^*) - \sigma_0((T\setminus J) \cup R) \leq 2/h$.
Observe that for a pair of candidates $c_1, c_2 \in l_a$, $d(c_1, c_2) \leq 2/h$ (see Figure~\ref{fig:eptas}).
Let $C_a = l_a \cap ((T \setminus J) \cup R)$.
Moreover, let $V_a \subseteq V$ be the set of voters which have a candidate in $C_a$ as their closest candidate in the committee $(T \setminus J) \cup R$.
Using the triangle inequality, we know that $d(V_a, c_a') \leq d(V_a, C_a) + 2/h$.
The above statement holds for all cells $l_a \in L'$.
Therefore, $\sigma_0((T^* \setminus J^*) \cup R^*) - \sigma_0((T\setminus J) \cup R) \leq 2/h$.

Note that our proof works for an arbitrary failing set $J$ including $J = \emptyset$.
This completes the proof of Lemma~\ref{lem-TTstar}.
\end{proof}

Based on the above observation, we solve the problem as follows.
We enumerate all maps $\chi\colon L \rightarrow \{0,1,\dots,h^4\}$ where $\chi(l_a)$ is the number of candidates from $l_a$ in the committee.
The total number of such maps is $h^{O(h^4)}$.
For each feasible map, i.e., $\chi$ satisfying $\chi(l_a) \leq |C \cap l_a|$ for all $a \in [h^4]$, we construct a fault-tolerant committee $T_\chi^*$ for $V_{i,j}$ by picking (arbitrarily) $\chi(l_a)$ candidates in $C \cap l_a$ for all $a \in [h^4]$ and including them in $T_\chi^*$.
For each constructed $T_\chi^*$, we compute a number $\widetilde{\sigma_f}(T_\chi^*)$ that approximates $\sigma_f(T_\chi^*)$ using the following lemma.
\begin{restatable} []{lemma}{score}
\label{lem-tilde}
Given $T_\chi^*$, one can compute a number $\widetilde{\sigma_f}(T_\chi^*)$ in $2^{O(h^4)}n^{O(1)}$ time such that $|\widetilde{\sigma_f}(T_\chi^*) - \sigma_f(T_\chi^*)| \leq 2/h$.
\end{restatable}
\begin{proof}
For a pair of candidates $c_i, c_j$ in a cell $l_i \in L$, we know $d(c_i, c_j) \leq 2/h$.
%Recall that for a pair of candidates $c_i, c_j \in C$ such that $c_i, c_j \in l_i$ for some $l_i \in L$, then $d(c_i, c_j) \leq 2/h$.
Since we want to compute the number $\tilde{\sigma_f}(T_{\chi}^*)$ within an absolute error of $2/h$ compared to the actual value, it is sufficient to only consider the failing sets for which either all or none of the candidates from a cell fail.
The total number of cells is $h^4$; therefore, we only need to consider at most $2^{O(h^4)}$ distinct failing cells (see Figure~\ref{fig:eptas}).
For each of these failing sets (say $J \subseteq C$), we can compute a best replacement committee $R$ in time $2^{O(h^4)}$ by either choosing one or zero candidates from each cell.
For each replacement, $\sigma_0(T_{\chi}^* \setminus J \cup R)$ can be computed in $O(nh^4)$ time.
Therefore, we can compute $\tilde{\sigma_f(T_\chi^*)}$ in time $2^{O(h^4)}n^{O(1)}$.
\end{proof}

Finally, we let $T_{i,j}$ be the smallest among all committees $T_\chi^*$ satisfying $\widetilde{\sigma_f}(T_\chi^*) \leq 1+4/h$, and we return it as our solution.
The running time of our algorithm is clearly $h^{O(h^4)} n^{O(1)}$.
The following lemma shows that our algorithm is correct.
\begin{restatable}[]{lemma}{eptasCorrectness}
\label{lem-eptas-correctness}
We have $\sigma_f(T_{i,j}) \leq 1+6/h$.
Furthermore, $|T_{i,j}| \leq |T|$ for any fault-tolerant committee $T$ for $V_{i,j}$ with $\sigma_f(T) \leq 1$.
\end{restatable}
\begin{proof}
The fact that $\sigma_f(T_{i,j}) \leq 1+6/h$ follows directly from our construction and Lemma~\ref{lem-tilde}.
Let $T$ be a fault-tolerant committee for $V_{i,j}$ with $\sigma_f(T) \leq 1$.
We consider two cases: $|T| > h^4$ and $|T| \leq h^4$.
First, assume $|T| > h^4$.
Define $\chi \colon L \rightarrow \{0,1,\dots,h^4\}$ by setting $\chi(l_a) = 1$ for all $a \in [h^4]$ with $C \cap l_a \neq \emptyset$, and $\chi(l_a) = 0$  whenever $C \cap l_a =\emptyset$.
Clearly, $|T_\chi^*| \leq h^4 < |T|$.
By Lemma~\ref{lem-TTstar}, $\sigma(T_\chi^*) \leq 1+2/h$.
Thus, $\widetilde{\sigma_f}(T_\chi^*) \leq 1+4/h$ by Lemma~\ref{lem-tilde}.
This further implies that $|T_{i,j}| \leq |T_\chi^*| < |T|$.

Now assume $|T| \leq h^4$.
Define $\chi \colon L \rightarrow \{0,1,\dots,h^4\}$ by setting $\chi(l_a) = |T \cap l_a|$ for all $a \in [h^4]$.
Clearly, $|T_\chi^*| = |T|$ and $|T_\chi^* \cap l_a| = |T \cap l_a|$ for all $a \in [h^4]$.
We show that $\sigma(T_\chi^*) \leq 1+2/h$.
Since $|T_\chi^* \cap l_a| = |T \cap l_a|$, for each $a$ pick a bijection $\pi_a \colon C \cap l_a \rightarrow C \cap l_a$ such that for all $x \in C \cap l_a$, $x \in T_\chi^*$ if and only if $\pi_a(x) \in T$.
Observe that the distance between $x$ and $\pi_a(x)$ is at most $2/h$ for all $x \in C \cap l_a$.
Combining all bijections $\pi_a$, we obtain a bijection $\pi \colon C_{i,j} \rightarrow C_{i,j}$ with the property that for all $x \in C_{i,j}$, the distance between $x$ and $\pi(x)$ is at most $2/h$, and $x \in T_\chi^*$ if and only if $\pi(x) \in T$.
Because of this bijection, it is obvious that $|\sigma_f(T_\chi^*) - \sigma_f(T)| \leq 2/h$ and in particular $\sigma_f(T_\chi^*) \leq 1+2/h$.
Thus, $\widetilde{\sigma_f}(T_\chi^*) \leq 1+4/h$ by Lemma~\ref{lem-tilde}.
This further implies that $|T_{i,j}| \leq |T_\chi^*| = |T|$.
\end{proof}

% \section{Conclusion and Open Problems} We introduce a novel fault-tolerance model and study its complexity on Euclidean elections under Chamberlin Courant rule (CC-rule). We present several exact and (near) optimal approximation algorithms. Our work suggests several new and exciting research directions for achieving fault-tolerance in committees, such as  (1) extending our results to other commonly used measures such as $k$-median or $k$-means \cite{jain1988algorithms}, (2) Investigating fault tolerance on the ordinal voter preference models for various well-known committee selection rules. Intuitively, it seems that dealing with fault-tolerance is \emph{computationally} easier when an optimal winner computation is polynomial time (we show this with $1d$ and single-peaked elections). It is interesting to further investigate this intuition on other voting scenarios \cite{elkind2017properties,faliszewski2019committee}.

\bibliography{references}

\newpage
\appendix

\end{document}